\documentclass[onecolumn,11pt]{IEEEtran}

\usepackage{verbatim}
\usepackage{amsmath}
\usepackage{amsthm}
\usepackage{amssymb}
\usepackage{mathrsfs}
\usepackage{color}
\usepackage{algorithm}
\usepackage{algorithmic}
\usepackage{bm}
\usepackage{multirow}

\usepackage{color}
\usepackage{soul}

\usepackage{graphicx}

\newtheorem{theorem}{Theorem}
\newtheorem{corollary}{Corollary}

\newtheorem{lemma}{Lemma}

\newtheorem{definition}{Definition}
\newtheorem{example}{Example}

\newcommand{\ket}[1]{\left| #1 \right\rangle}
\newcommand{\bra}[1]{\left\langle #1\right |}

\def\FF{\mathbb{F}}

\def\sX{\mathsf{X}}
\def\sZ{\mathsf{Z}}

\newcommand{\cH}{{\mathcal{H}}}
\newcommand{\cW}{{\mathcal{W}}}

\def\Tr{\mathop{\rm Tr}\nolimits}
\def\tr{\mathop{\rm tr}\nolimits}
\def\Ed{E_P}

\def\ba{\vec{a}}
\def\bb{\vec{b}}

\def\inc{\mathop{\mathbf{I}}\nolimits}
\def\cin{\mathop{\mathbf{CI}}\nolimits}
\def\qin{\mathop{\mathbf{QI}}\nolimits}

\def\Label#1{\label{#1}\ [\ \text{#1}\ ]\ }
\def\Label{\label}

\begin{document}
\sloppy

\title{Single-Shot Secure Quantum Network Coding for General Multiple Unicast Network with Free One-Way Public Communication}

\author{
\thanks{
The material in this paper was presented in part at 
the 10-th International Conference on Information Theoretic Security  (ICITS 2017), Hong Kong (China),  November 29 - December 2, 2017 \cite{ICITS}.
}
Go Kato\thanks{The first author is with NTT Communication Science Laboratories, NTT Corporation, Japan, e-mail:kato.go@lab.ntt.co.jp.},
Masaki Owari\thanks{The second author is with Department of Computer Science, Faculty of Informatics, Shizuoka University, Japan, e-mail:masakiowari@inf.shizuoka.ac.jp.},       
and 
Masahito~Hayashi~\IEEEmembership{Fellow,~IEEE}
\thanks{The third author is with the Graduate School of Mathematics, Nagoya University, Japan. He is also with 
Shenzhen Institute for Quantum Science and Engineering,
	Southern University of Science and Technology,
	Shenzhen,
	518055, China,
	Center for Quantum Computing, Peng Cheng Laboratory, Shenzhen 518000, China,
and 
the Centre for Quantum Technologies, National University of Singapore, Singapore, e-mail:masahito@math.nagoya-u.ac.jp}}

\maketitle

\begin{abstract}
It is natural in a quantum network system that multiple users intend to  
send their quantum message to their respective receivers, which is  
called a multiple unicast quantum network. We propose a canonical method  
to derive a secure quantum network code over a multiple unicast quantum  
network from a secure classical network code. Our code correctly  
transmits quantum states when there is no attack. It also guarantees the  
secrecy of the transmitted quantum state
even with the existence of an attack when the attack satisfies a certain  
natural condition.
In our security proof, the eavesdropper is allowed to modify wiretapped  
information dependently on the previously wiretapped messages.
Our protocol guarantees the secrecy by utilizing one-way classical information  
transmission (public communication) in the same direction as the quantum network although the  
verification of quantum information transmission requires two-way  
classical communication.
Our secure network code can be applied to several networks including the  
butterfly network.
\end{abstract}

\section{Introduction}
In order to realize quantum information processing protocols to overwhelm the 
conventional information technologies among multiple users,
it is needed to build up a quantum network system among multiple users.
For example, 
various quantum protocols, e.g., 
quantum blind computation \cite{BFK,MF},
quantum public key cryptography \cite{KK},
and quantum money \cite{Wiesner} require 
the transmission of quantum states.
To meet the demand, the paper \cite{Hayashi2007} initiated the study of quantum network coding with the butterfly network as a typical example.
Under this example, the paper \cite{PhysRevA.76.040301} clarified the importance of prior entanglement in a quantum network code
by proposing a network code, which was experimentally implemented recently\cite{Lu}.
Kobayashi et al. \cite{Kobayashi2009} discussed a method for generating GHZ-type states via quantum network coding.
Leung et al. \cite{Leung2010} investigated several types of networks 
when classical communication is allowed.
Based on these studies, Kobayashi et al. \cite{Kobayashi2010} made a code to transmit quantum states based on a linear classical network code.
Then, Kobayashi et al. \cite{Kobayashi2011} generalized the result to the case with non-linear network codes. 
These studies \cite{Kobayashi2009,Leung2010,Kobayashi2010,Kobayashi2011,JFM2011} clarified that quantum network coding is needed among multiple users for efficient transmission of the quantum states over a quantum network.
However, these existing studies did not discuss the security for 
quantum network codes
when an adversary attacks the quantum network.

Since the improvement of the security is one of the most essential requirements for developing quantum networks,
the security analysis is strongly required for quantum network codes. 
Indeed, it is possible to check the security in these existing methods by verifying the non-existence of the eavesdropper.
However, the verification requires us to repeat the same quantum state transmission several times as well as two-way classical communication.
Hence, it is impossible to guarantee the security under 
a single transmission in the simple application of these existing methods.
Therefore, it is needed to propose a quantum network code that guarantees its security. 
That is, our aim is a natural extension of classical secure network coding.

On the other hand, for a classical network, 
Ahlswede et al.\cite{ACLY} started the study of network coding. 
Then, Cai et al. \cite{Cai2002} initiated to address the security of network code, 
and pointed out that the network coding enhances the security.
Currently, many papers \cite{Cai06,Bhattad2005,Liu2007,Rouayheb2007,Harada2008,Ho2008,Jaggi2008,Nutman2008,Yu2008,Cai2011,Cai2011a,Matsumoto2011,Matsumoto2011a} have already studied the security for network codes.  
In these studies, the security was shown against wiretapping on a part of the channels.
Hence, it is strongly needed to propose a quantum network code whose security is guaranteed under a similar setting.
In the previous paper \cite{OKH}, we initiated a study of the security of quantum network codes, and constructed a quantum network code on the butterfly network which is secure against any eavesdropper's attack on any one of quantum channels on the network.
In fact, after the conference version \cite{ICITS} of this paper,
several studies \cite{SH18-1,SH18-2,HS19} investigated the security for the quantum network code
when an adversary attacks the quantum network.
However, they did not discuss a method for converting 
an existing classical network code to a quantum network code.
Our method can be universally applied to any classical network code as follows.

To see our contribution, we explain the characteristics of a quantum network.
Studies on classical network coding have most often discussed the unicast setting, in which,
we discuss the one-to-one communication via the network.
Even in the unicast setting, there are many examples of network codes that overcome the routing, as numerically reported in \cite[Section III]{CF}.
However, in the quantum setting, it is not easy to find such an example in the unicast setting.
As another formulation, 
studies on classical network coding often focus on the multicast setting,
in which one sender sends information to multiple receivers.
However, no-cloning theorem prohibits a straightforward extension of classical multicast network coding to quantum multicast network coding, even though there exists various types of quantum multicast communication protocols utilizing classical multicast network codes \cite{Kobayashi2009,KOM,KOM2,HO}.
Hence, we discuss the multiple-unicast setting, which has multiple pairs composing of senders and receivers, 
since we can construct a problem setting of quantum multiple-unicast network coding as a straightforward extension of the problem setting of classical multiple-unicast network coding.   
In addition, the multiple-unicast setting has not been well examined even in the classical case, 
i.e., it has been discussed only in a few papers such as Agarwal et al. \cite{Agarwal} with the classical case.

In this paper, we generally construct a quantum linear network code 
in the multiple-unicast setting whose security is guaranteed.
Our code is canonically constructed from a classical linear network code in the multiple-unicast setting, 
and it certainly transmits quantum states when there is no attack.
Our main issue is the secrecy of the transmitted quantum states 
when Eve attacks only edges in the subset $E_A$ of the set of edges of the given network.
That is, we show the secrecy of our quantum network code
when the secrecy and recoverability of the corresponding classical network are shown against Eve's attack on the subset $E_A$ of edges.
That is, we clarify the relation between quantum secrecy and the pair of classical secrecy and classical recoverability in the network coding.
We also give several examples of such secure quantum network codes.
Indeed, it is not so easy to satisfy this condition for the corresponding classical network.
Hence, we allow several nodes in the network to share common randomness, which is called shared randomness, and assume that Eve priorly does not have any information about this shared randomness.
Since a quantum channel is much more expensive than a classical public channel,
we assume that the classical one-way public channel can be freely and unlimitedly used
from each node only to terminal nodes which are in the directions of the subsequent quantum communications.
Under this assumption, 
the transmission of a quantum state from a source node to the corresponding terminal node
is equivalent to sharing a maximally entangled state via quantum teleportation \cite{PhysRevLett.70.1895}.
Hence, we show the reliability of the  transmission by proving that an entangled state can be shared by sending entanglement halves from source nodes.
Our general construction covers the previous code for the butterfly network in \cite{OKH}.

Here, we emphasize the difference between our offered security from 
the conventional quantum security like quantum key distribution (QKD),
which essentially verifies the noiseless quantum communication.
In QKD, for this verification, we need two-way classical communication, which enables us to verify the non-existence of the eavesdropper and to ensure the security.
However, our analysis can guarantee the security only with
one-way classical communication
because we assume that the eavesdropper wiretaps only a part of the channels.
Also, the verification in QKD can be done under an asymptotic setting
with repetitive use of quantum communications.
In contract, our security analysis holds even with the single-shot setting
without such repetitive use.

The remaining part of this paper is organized as follows.
Section \ref{sec:Preliminary} prepares several pieces of knowledges for secure classical network coding including 
secrecy and recoverability.
Section \ref{S6} provides our general construction of secure quantum network coding and 
shows the secrecy theorem.
Section \ref{S7} discusses several additional examples of secure quantum network coding.
Appendix \ref{AL} gives several lemmas used in these examples.
Appendix \ref{AConst} gives the precise constructions of the matrices appearing in the main body.

\section{Preparation from secure classical network coding}\Label{sec:Preliminary}
In this section, we introduce classical network coding and its secrecy
and recoverability
analysis which is necessary for analyzing the security 
of
 the derived 
quantum network codes in the next section. 

\subsection{Classical linear multiple-unicast network coding}
The quantum multiple-unicast network codes can
 be derived from any classical  linear multiple-unicast network codes with shared-secret randomness, where 
the linearity condition is imposed for the operations on all the nodes.
In the classical setting of the network coding, 
the network is
expressed by
 a directed graph $(\tilde{V}, \tilde{E})$. 
The set of vertices $\tilde{V}$
indicates
   the set of nodes which are the senders and  receivers of communications.
The set of edges $\tilde{E}$ 
indicates
 the set of the communication channels,
i.e., the set of packets.
When a single character in $\FF_q$ is transmitted from a vertex $u \in \tilde V$ to another vertex $v \in \tilde V$ via a channel
in the classical network code, the channel is 
indicated by
 $(u,v) \in \tilde{E}$ in the directed graph. Note that
 $\FF_q$ is the finite field whose order $q$ is a prime power.

The purpose of 
the multiple-unicast network coding is that
 the nodes cooperatively   transmit  the $n$ messages from 
a part of the nodes (called source nodes) to other part of  nodes (called terminal nodes).
For individual message, 
the  pair of the source node and the terminal node are predefined.
A single source or terminal node may appear multiple times in the set of 
the  
pairs.
In other words,
 a source node may be required to send messages to plural terminal nodes, and plural source nodes may be required to send messages to an identical terminal node. 
As you can find,
 our setting includes the unicast setting as well. 
In our classical network code setting, we consider the situation that part of communications are eavesdropped by Eve. 
In order to make the code secure, i.e. to prevent the leakage of information correlated to the messages,  
$n'$ shared-secret randomnesses are used. 
The nodes which use any of the randomnesses  are called shared-randomness nodes.
To make the following discussion clearer, we
  denote the sets of source nodes, terminal nodes, and shared-randomness nodes by $V_{S}$, $V_T$, and $V_{SR}$.

In fact, the notation defined above is not enough to analyze the network code systematically, especially in the case of the derived quantum network code. 
Therefore, we will extend the structure of the network.

\subsubsection{ Definition of sets which characterize the extended network}
As an extension of  the network,
 we virtually 
 introduce 
additional  
input vertices, output vertices, and shared-randomness vertices,
such that 
there is one-to-one correspondence between the $j$-th message and the pair of the input vertex $i_{j}$ and 
the output vertex $o_{j}$ for $1 \leq j\leq n$, and there is  one-to-one correspondence between the 
$j$-th shared-secret randomness and the shared-randomness vertex $r_j$ for $1 \leq j\leq n'$.
In the following,
 we
 denote the sets of input vertices, output vertices, and shared-randomness vertices by 
 $V_I=\{i_{1},\cdots,i_{n}\}$, $V_O=\{o_{1},\cdots,o_{n}\}$, and $V_R=\{r_{1},\cdots,r_{n'}\}$,
  respectively.
Now, we give the set of vertices for the extended network 
 as   
$V:= \tilde{V}\cup V_I\cup V_O\cup V_R$, where these sets have no intersection.

Next, we virtually add input edges, output edges, and shared-randomness edges which connect between virtual vertices and the nodes in $\tilde V$ 
so as to satisfy the following conditions:
Any input vertex is connected only by an input edge to the source node which possesses  the corresponding message initially in the classical  network code. 
Any output vertex is connected only by an output edge from the terminal node which 
  receives the corresponding message finally in the  network code.
 And, any shared-randomness vertex $r_j$ is connected only by $l_j$ shared-randomnesses edges to  
all the shared-randomness nodes where the corresponding shared-secret randomness is distributed initially in the  network code.
From now on,
 we denote the sets of input edges, output edges, and shared-randomness edges by $E_I$, $E_O$, and $E_R$, respectively.
From these definitions, we know that
 $|E_I|=|E_O|=n$
and 
 the number $|E_{R}|$ of shared-randomness edges is $l:=\sum_{j=1}^{n'}l_j$.
Now, we give the set of edges for the extended network as   
$E:= \tilde{E}\cup E_I\cup E_O\cup E_R$.
Note that, these sets have no intersection, so, $|E|= N+2n+l$
where $N$ is the number of edges $|\tilde{E}|$ for the original network.
In the following, we doesn't distinguish the edges and the corresponding communication channels.
To clarify what we have defined, we show typical relations for the set defined above:
\begin{align}
V_{S} &
 =\left\{ v\in V|\exists u\in V_{I}\ s.t.\ \left(u,v\right)\in E\right\},\nonumber \\
V_{SR} &
 =\left\{ v\in V|\exists u\in V_{R}\ s.t.\ \left(u,v\right)\in E\right\},\nonumber \\
V_{T} &
 =\left\{ u\in V|\exists v\in V_{O}\ s.t.\ \left(u,v\right)\in E\right\},
\nonumber \\
E_{I} & =\left\{ \left(u,v\right)\in E|u\in V_{I},v\in V\right\} ,\nonumber \\
E_{O} & =\left\{ \left(u,v\right)\in E|u\in V,v\in V_{O}\right\} ,\nonumber \\
E_{R} & =\left\{ \left(u,v\right)\in E|u\in V_{R},v\in V\right\} .\Label{eq:def E-I E-O E-R}
\end{align}
Some numbers
defined above
 are summarized in Table \ref{T1} for convenience.
\begin{table}[htpb]
  \caption{Characteristic numbers of the network coding.
Notations undefined here will be defined later.}
\Label{T1}
\begin{center}
  \begin{tabular}{|c|l|} 
\hline
\multirow{4}{*}{$n$} & No. of input edges. $|E_I|$ \\
&  No. of output edges. $|E_O|$  \\
&  No. of input vertices. $|V_I |$  \\
&  No. of output vertices. $|V_O|$  \\
\hline
$l $ & No. of shared-randomness edges. $|E_R|$ \\
\hline
$n'$ & No. of shared-randomness. $|V_R |$\\
\hline
$N$ & No. of
   edges
 in the original network. $|\tilde{E}|$\\
\hline
$h$ & No. of edges attacked by Eve. $|E_A|$\\
\hline
$h'$ & No. of protected edges.  $|\Ed|$\\
\hline
  \end{tabular}
\end{center}
\end{table}

\subsubsection{ Definition of maps which characterize the network code}

To identify the ordering of the channels,
 we define a map $\mathbf{e}$ from 
$\left\{ 1,\cdots,N+2n+l\right\} $
to $E$ as follows.
For $1\leq j\leq n$, 
$\mathbf{e}\left(j\right)$ is an input edge
going out from an input vertex $i_{j}$.
For $1\leq j\leq n'$, 
$\mathbf{e}\left(n+1+\sum_{k=1}^{j-1}l_{k}\right),\cdots,
\mathbf{e}\left(n+\sum_{k=1}^{j}l_{k}\right)$
are sheared-randomness edges
going out from 
a sheared-randomness vertex
$r_{j}$.
 $\mathbf{e}\left(n+l+1\right),\cdots,\mathbf{e}\left(N+n+l\right)$
are edges in the directed graph $(\tilde{V},\tilde{E})$ which is  originally defined  in 
the classical
 network coding.
For $1\leq j\leq n$, 
 $\mathbf{e}\left(N+n+l+j\right)$ is an output
edge going into an output vertex $o_{j}$.
We consider that, if $j<k$, the channel $\mathbf{e}(j)$ is used before the channel $\mathbf{e}(k)$ is used, 
and we regard $j$ as the time $t$ when the channel $\mathbf{e}(j)$ is used.

For an edge 
$\mathbf{e}\in E$,
 we denote its input and output vertices by 
$\mathbf{v}_{I}\left(\mathbf{e}\right)$
and $\mathbf{v}_{O}\left(\mathbf{e}\right)$.
The definition means that the relation 
$\mathbf{e}=\left(\mathbf{v}_{I}\left(\mathbf{e}\right),\mathbf{v}_{O}\left(\mathbf{e}\right)\right)$
holds.
Note that: though we use the same character $\mathbf e$ for  both a ``edge'' variable   in $E$ and a function defined in the previous paragraph, 
we don't mention about it below if it is easy to identify which it means. 
Next, we define a map from integers $j$
to
 the set
of natural numbers identifying the edges that have transmitted their
contents to the vertex $\mathbf{v}_{I}\left(\mathbf{e}\left(j\right)\right)$
before the 
time $t=j$:
\begin{equation}
\inc\left(j\right):=\left\{ k\in\mathbb{N}|k<j,\ \exists v\in V,\ s.t.\ \mathbf{e}\left(k\right)=
\left(v,\mathbf{v}_{I}\left(\mathbf{e}\left(j\right)\right)\right)\right\} ,\Label{eq:def in j}
\end{equation}
which enables us to make expressions simple.

We consider that,
at the $j$-th channel, 
a random variable $Y_j$ is imputed as a transferred content in the  network.
Therefore, at the time $t=j$, $Y_{j}$ is generated at the node $\mathbf{v}_{I}\left(\mathbf{e}\left(j\right)\right)$,
and, transmitted it via the channel $\mathbf{e}\left(j\right)$.
Immediately after
the time, the value
 is received by the vertex $\mathbf{v}_{O}\left(\mathbf{e}\left(j\right)\right)$
 if there is no disturbance.
 In fact, we will consider the case that Eve may disturb the contents of part of channels later.
Note that,
in the case of $1\leq  j\leq n$  or $N+n+l< j\leq N+2n+l$,
we consider that the $j$-th channel virtually transfers a message  which is initially occupied at the corresponding source node or is finally reconstructed at the corresponding terminal node respectively. Furthermore,  in the case of  $n< j\leq n+l$,
we consider that the $j$-th channel virtually transfers a sheared-secure randomness 
which is sheared  at the corresponding sheared-randomness nodes.
This interpretation indicates that 
 $\{Y_j\}_{ j\in\{1,\cdots, n\}}$ are the  messages,
 and
$\{Y_j\}_{j\in\{n+1,\cdots , n+l\}}$ are the sheared-secret randomnesses,

In order to fix the 
classical linear network codes, 
the rest of work we have to do is to give the way to generate the random variables $Y_j$  which is transferred by the $j$-th  channel $\mathbf{e}\left(j\right)$ for $n+l< j\leq N+2n+l$.  Recall that
 we impose the linearity condition on the operations on all the nodes, and, 
only the random variables $\{Y_k\}_{k \in \inc\left(j\right)}$ can be used for generating $Y_j$
 at the vertex $\mathbf{v}_{I}(\mathbf{e}\left(j\right))$.
Therefore, 
  $Y_{j}$ can be evaluated as a linear combination 
   $\sum_{k\in \inc\left(j\right)}\theta_{j,k}Y_{k}$ for appropriate constants $\theta_{j,k} \in \FF_q$.
  We can easily check that 
 the set 
 $\left\{ \theta_{j,k}\right\} _{j\in\left\{n+l+1,\dots,\left|E\right|\right\}, k\in \inc\left(j\right)}$
completely 
identifies
 the 
linear multiple-unicast coding on the given network.
For convenience, we define $\theta_{j,k}=0$ for 
$k\nsubseteq \inc\left(j\right)$ so that we have
\begin{equation}
Y_{j}=
\sum_{k\in \inc\left(j\right)}\theta_{j,k}Y_{k}=
\sum_{k<j}\theta_{j,k}Y_{k}.\Label{eq:y(i)=00003Dsum a-ij y(j)}
\end{equation}
for $n+l+1\leq j\leq N+2n+l$.
Note that, the above relation doesn't hold if there is disturbance on the channels, e.g. attacks by Eve, 
because, in that case, no one guarantee that the received content from the edge $\mathbf e(j)$ is equal to the sent content into the      the edge $\mathbf e(j)$, i.e. $Y(j)$.

\begin{figure}[tbh]
\begin{center}
\includegraphics[scale=0.6]{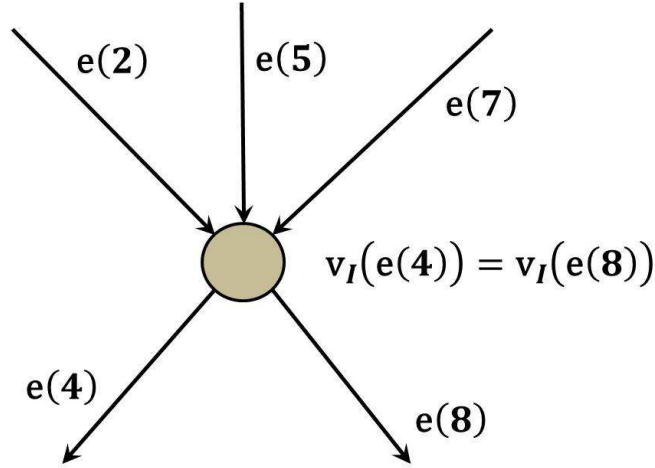}
\caption{Example of a vertex and connecting edges}
\Label{fig:vertex-edges}
\end{center}
\end{figure}

\begin{example}
As an example,  a local structure for a network defined above is depicted in the Figure \ref{fig:vertex-edges}, i.e.  a vertex and
connecting edges. The edges $\mathbf{e}\left(2\right)$, $\mathbf{e}\left(5\right)$,
and $\mathbf{e}\left(7\right)$ go into the vertex and the edges $\mathbf{e}\left(4\right)$
and $\mathbf{e}\left(8\right)$ go out from the vertex. 
Both  $\mathbf{v}_{I}\left(\mathbf{e}\left(4\right)\right)$
and 
$\mathbf{v}_{I}\left(\mathbf{e}\left(8\right)\right)$ indicate
the vertex.
At 
the time $t=4$,
 the content from $\mathbf{e}\left(2\right)$ has
arrived, but the contents from $\mathbf{e}\left(5\right)$ and $\mathbf{e}\left(7\right)$ have not yet. 
Since $\mathbf{e}\left(1\right)$
and $\mathbf{e}\left(3\right)$, which do not appear in Figure \ref{fig:vertex-edges},
do not connect to $\mathbf{v}_{I}\left(\mathbf{e}\left(4\right)\right)$,
the operation on $\mathbf{v}_{I}\left(\mathbf{e}\left(4\right)\right)$
is determined by $\theta_{4,2}$
only, and $\left\{ \theta_{4,j}\right\} _{j<4}$
can be written as 
\[
\left\{ \theta_{4,j}\right\} _{j< 4}:=\left(0,\theta_{4,2},0\right).
\]
Similarly, at
the time $t=8$, all the
  contents from $\mathbf{e}\left(2\right)$,
$\mathbf{e}\left(5\right)$, and $\mathbf{e}\left(7\right)$ have been received
at $\mathbf{v}_{I}\left(\mathbf{e}\left(8\right)\right)$. Thus, the content
sent by $\mathbf{e}\left(8\right)$ can be written as $\sum_{j<8}\theta_{8,j}Y_{j}$, 
where $Y_{j}$ is content received from $\mathbf{e}\left(j\right)$, 
and 
$\left\{ \theta_{8,j}\right\} _{j<8}=\left(0,\theta_{8,2},0,0,\theta_{8,5},0,\theta_{8,7}\right)$. 
\end{example}

Due to the linear structure given in \eqref{eq:y(i)=00003Dsum a-ij y(j)},
the random variables $Y_j$ is given as a linear combination of
the messages $\Vec{A}:=(A_{1},\cdots,A_{n})$ 
given
 in the input vertices 
and the shared-secure-random variables $\Vec{B}:=(B_{1},\cdots,B_{n'})$ generated at the shared-randomness vertices virtually.
For simplicity, combining these random variables,
we define the random vector $\vec{X}':=(\vec{A},\vec{B})=(X_1',\cdots, X_{n+n'}')$.
From the constants $\left\{ \theta_{j,k}\right\} _{j\in\left\{1+n+l,\dots,\left|E\right|\right\},k\in \inc\left(j\right)}$,
we can  uniquely construct  an $\FF_q$-valued
$(N+2n+l)\times (n+n') $ matrix
 $M_0$ whose $(j,k)$ element is $m_0\left(j,k\right)$
  such that
\begin{equation}
Y_{j}=\sum_{k=1}^{n+n'}m_0\left(j,k\right)X_{k}',
\Label{H204}
\end{equation}
if there is no disturbance.
The concrete construction of $M_0$ is given in Appendix \ref{a1}.
Since $\mathbf{e}\left(j\right)$
 is an output edge for $N+n+l+1\le j\le N+2n+l$, the corresponding 
elements
 $m_{0}\left(j,k\right)$ must satisfy 
\begin{align}
\left\{ m_{0}\left(N+n+l+j',k'\right)\right\} _{k'=1}^{n+n'} & =(\vec 0_{j'-1},1,\vec 0_{n+n'-j'})
\Label{eq:M(n+l+N+c k) for 1<c<n}
\end{align}
for $1\le j' \le n$. Note that  we use the notation $\vec 0_j:=(\overbrace{0,\cdots,0}^{j})$.
Rigorously writing,  we can define a multiple-unicast network code by the following condition: 
\begin{definition}
\Label{Def A-network-code}
Let $(E,V)$ be a directed network with ordered edges, $n$ and $l$ be constants which satisfies $2n+l\leq |E|$, and $I(j)$ be a map defined by Eq. (\ref{eq:def in j}).
A network code $\left\{ \theta_{j,k}\right\} _{j\in\left\{1+n+l,\dots,\left|E\right|\right\},k\in \inc\left(j\right)}$
is called a multiple-unicast network code 
if the coefficients $\left\{ m_{0}\left(j,k\right)\right\} _{j,k}$
satisfy Eq.\eqref{eq:M(n+l+N+c k) for 1<c<n}, where the coefficients
$\left\{ m_{0}\left(j,k\right)\right\} _{j,k}$ are defined by Eq.
\eqref{H204}.
\end{definition}

\subsection{Secrecy of classical multiple-unicast network code}

In this subsection, we analyze  the secrecy of the classical network code. The analysis 
 is necessary to derive the main results regarding quantum network codes. 
Although there are a lot of existing works on secrecy of classical network coding \cite{Cai2002,Bhattad2005,Liu2007,Rouayheb2007,Harada2008,Ho2008,Jaggi2008,Nutman2008,Yu2008,Cai2011,Cai2011a,Matsumoto2011,Matsumoto2011a},
they don't discuss the case when an adversary called ``Eve'' 
disturbs the contents on the part of channels as well as she wiretaps  the part of channels.
Only the paper \cite{HOKC} discusses such an adversary, though its analysis is limited to the unicast case.

\subsubsection{ Definitions related to Eve's attack}

We define $E_{A}\subset\tilde{E}$ 
 as the set of edges attacked by Eve, 
and $h$ is 
 the size of 
the set,
 \emph{i.e., $h:=\left|E_{A}\right|$}, respectively.
 Note that, since all the edges in the set $E\backslash \tilde E$ are virtual ones, Eve can't access the edges.
Eve is assumed to be able to eavesdrop and disturb the contents  on all the channels in $E_A$. 
Eve also knows the network structure, i.e., the topology of network and all the coefficients $\{\theta_{j,k}\}_{j,k}$.
In order to make expressions simply, we define a strictly increasing function $\varsigma\left(j\right)\in\mathbb{N}$
so that $E_{A}$ can be written as $E_{A}=\left\{ \mbox{\ensuremath{\mathbf{e}\left(\varsigma\left(j\right)\right)}}\right\} _{j=1}^{h}$.
That is, 
the target of  the $j$-th Eve's attack is
the edge $\mathbf{e}\left(\varsigma\left(j\right)\right)$.
In order to analyze such a situation, we introduce other random variables:
the wiretapped 
 random variable $Z_j:=Y_{\varsigma\left(j\right)}$
from the communication identified by  the edge $ \mathbf{e}\left(\varsigma\left(j\right)\right)$, and the  injected
     random variable $C_j$ to the vertex 
$ \mathbf{v}_{O}\left(\mathbf{e}\left(\varsigma\left(j\right)\right)
\right)$ instead of $Z_j$.
In order to simplify the following discussion, we define two random variable vectors:
 $\vec{C}:=(C_1, \ldots, C_h)$ 
and
 $\vec{X}:=(\vec{A},\vec{B},\vec{C})=(X_1,\cdots ,X_{n+n'+h})$.
Due to the linear structure of the network, 
there uniquely exists an $\FF_q$-valued
$(N+2n+l)\times (n+n'+h) $ matrix $M$ 
  whose $(j,k)$ element is $m\left(j,k\right)$
 satisfying 
 that 
the input information $Y_{j}$ of the edge $\mathbf{e}(j)$ can be expressed by 
\begin{equation}
Y_{j}=\sum_{k=1}^{n+n'+h}m\left(j,k\right)X_{k}
\Label{eq:sec classical subsec security y i =00003D M i k x k},
\end{equation}
 when the contents on the edges $E_A$ are disturbed  by  Eve.
The concrete construction of $M$ is given in Appendix \ref{a2}.
 
When we name 
 the received content from the edge $\mathbf{e}(j)$ 
  as
 $Y_j'$, the random variable can be defined as 
\begin{align}
Y_j':=\left\{
\begin{array}{ll}
C_k & \hbox{when there exists }j \hbox{ satisfying }j=\varsigma\left(k\right)\\
Y_j & \hbox{otherwise}.
\end{array}
\right.
\end{align}
 We can easily define 
 the $\FF_q$-valued
$(N+2n+l)\times (n+n'+h) $ matrix $M'$ 
 which gives 
$Y_j'$ from $\vec X$ as
\begin{equation}
Y_{j}'=\sum_{k=1}^{n+n'+h}m'\left(j,k\right)X_{k},
\Label{eq:def y'}
\end{equation} 
where $m'\left(j,k\right)$ is a $(j,k)$ elements of the matrix $M'$.
That is
\begin{align}
m'\left(j,k\right) 
:=&\left\{
\begin{array}{ll}
\delta_{k,n+n'+j'} & \hbox{when there exists }j' \hbox{ satisfying }j=\varsigma\left(j'\right)\\
m\left(j,k\right) & \hbox{otherwise}.
\end{array}
\right.\Label{eq:M' ik =00003D delta M ik}
\end{align}

From now on, we fix the set of edges where Eve attacks, i.e. $E_A$, and the network code, i.e. $\{\theta_{j,k}\}_{j,k}$. That means,  
the matrices $M_0$, $M$, $M'$, and the maps $\varsigma$, $\mathbf{e}$, $\mathbf{v}_{I}$, $\mathbf{v}_{O}$
 are fixed.

\subsubsection{ Categorization of Eve's attack}

In order to reduce complicate Eve's attack into simple one,
we
 categorize her attack into three types: \emph{simple attack, deterministic
attack} and \emph{probabilistic attack}:\\
{\bf Simple attack}: A simple attack is an attack in which Eve just
deterministically chooses her injecting value $
\vec{C}=\left\{ X_{n+n'+j}\right\} _{j=1}^{h}$ as a constant.
Therefore, the injected value is independent from the 
wiretapped
values $\left\{ Z_j\right\} _{j=1}^{h}$.
\\
{\bf Deterministic attack}: 
A deterministic attack is defined by
a set of functions $\left\{ g_{j}\right\} _{j=1}^{h}$:
\[
g_{j}:\mathbb{F}_{q}^{j-1}\rightarrow\mathbb{F}_{q},
\]
where $g_j$ is not restricted to a linear function.
The function $g_j$ gives Eve's $j$-th injected value $C_j=X_{n+n'+j}$
generated from the wiretapped variables $\{Z_k\}_{k=1}^{j-1}$ in her hand
 as
\begin{align*}
C_j =
g_{j}\left(
\{Z_k\}_{k=1}^{j-1}
\right) 
 &
  =g_{j}\left(\left\{ \sum_{k'=1}^{n+n'+h}m\left(\varsigma\left(k\right),k'\right) 
X_{k'}\right\}_{k=1}^{j-1}\right).
\end{align*}
This attack is a special case of the {\it causal strategy } defined in the paper \cite{HOKC}.
We write the set of all deterministic attacks as $\mathcal{G}$, i.e. all the set of $\left\{ g_{j}\right\} _{j=1}^{h}$.
Note that a simple attack is also a deterministic attack.

{\bf Probabilistic attack}: A probabilistic attack is an attack
in which Eve probabilistically chooses one of the deterministic attacks
$\left\{ g_{j}\right\}_{j=1}^h$ and applies it.
Hence, a probabilistic attack is determined by a probability distribution
$P_{\mathbf{G}}\left(
\left\{ g_{j}\right\}_{j=1}^h
\right)$ on
the set of all deterministic attacks $\mathcal{G}$, where $\mathbf{G}$
is the corresponding random variable. Note that a deterministic attack
$\left\{ g_{j}\right\}_{j=1}^h$
 is a special probabilistic attack
whose probability distribution satisfies 
$P_{\mathbf{G}}\left(
\left\{ g_{j}\right\}_{j=1}^h
\right)=1$
and $P_{\mathbf{G}}\left(
\left\{ g'_{j}\right\}_{j=1}^h
\right)=0$
for
 any other 
  deterministic attack 
$
\left\{ g'_{j}\right\}_{j=1}^h
$.

Note that, even in the case of probabilistic attack,  the set of edges where Eve attacks, i.e. $E_A$,  is fixed.

\subsubsection{ Reduction of complex Eve's attacks into simple ones}

First, we consider the deterministic attack.
In this case,
 any attack
 can be reduced to 
  a simple 
  attack with $\vec{C}=\vec 0=(0,0,\cdots)$, i.e.
  for any deterministic attack, there is a simple attack with $\vec{C}=\vec 0$ where 
  Eve  can get the same information with both strategies.
The reason is as follows. For the original deterministic attack $\{g_j\}_{j=1}^h$, 
Eve's information is given as 
$\left\{ Z_j\right\} _{j=1}^{h}$. In the case of the simple attack, i.e. $\vec{C}=\vec 0$,
 we denote Eve's information  by
$\left\{ \tilde{Z}_{j}\right\} _{j=1}^{h}$ .
Due to the linearity of the network, we have
\begin{align}
Z_j = \tilde{Z}_j
+ \sum_{k=1}^{h}m(\varsigma\left(j\right),n+n'+k) g_k(\left\{ Z_{k'}\right\}_{k'=1}^{k-1}),
\label{G_eq:1}
\end{align}
for $1\leq j\leq h$. As you can check, all the elements $m(j,k)$ are defined by the network code and the set of edges where Eve attacks.
 This fact guarantees  that we can solve the eq. ($\ref{G_eq:1}$)
with respect to $\{Z_k\}_{k=1}^j$.
This fact can be rewritten as the following lemma.
\begin{lemma}[\protect{\cite[Theorem 1]{HOKC}}]\Label{L2-4}
Any deterministic attack 
can be reduced to a simple attack with $\vec{C}=\vec 0_h$.
Since any probabilistic attack is given as a probabilistic mixture of 
deterministic attacks,
it can also be reduced to the simple attack with $\vec{C}=\vec 0_h$.
\end{lemma}
For reader's convenience, we summarize all the random variables we defined
 in Table \ref{T1-1}
\begin{table}[htpb]
  \caption{Random variables of the network coding.}
\Label{T1-1}
\begin{center}
  \begin{tabular}{|c|l|} 
\hline
$ A_j$ & The random variable of the $j$-th message\\
\hline
$ B_j$ & The random variable of the $j$-th shared number\\
\hline
$C_j$ & The random variable injected at the end of edge $\mathbf{e}(\varsigma\left(j\right))$\\
\hline
$X_j $ & An alias of the variable $A_j$, $B_{j-n}$, or $C_{j-n-n'}$
\\
\hline
$Y_j$ & The random variable inputted to the edge $\mathbf{e}(j)$\\
\hline
$Y_j'$ & The random variable outputted from the edge $\mathbf{e}(j)$\\
\hline
$Z_j$ & The random variable wiretapped at the edge $\mathbf{e}(\varsigma\left(j\right))$
\\
\hline
$\tilde Z_j$ &  The random variable wiretapped at the edge $\mathbf{e}(\varsigma\left(j\right))$
under the virtual condition $\vec C=\vec 0_h$
\\
\hline
  \end{tabular}
\end{center}
\end{table}

\subsubsection{ Security analysis for  Eve's attack on $E_A$}

For given $E_{A}$ and the function $\varsigma$,
we define a $h\times\left(n+n'+h\right)$
matrix $M_{\varsigma}$ whose elements are given by 
$\left\{ m\left(\varsigma\left(j\right),k\right)\right\}_{1\le j\le h, 1\le k \le n+n'+h}$.
We further define submatrices of $M_{\varsigma}$ as $M_{\varsigma}=\left(M_{\varsigma,1},M_{\varsigma,2},M_{\varsigma,3}\right)$
where the sizes of $M_{\varsigma,1}$, $M_{\varsigma,2}$, and $M_{\varsigma,3}$
are $h\times n$, $h\times n'$, and $h\times h$, respectively. 
For $ \vec{x}=\left(\vec{a},\vec{b},\vec{c}\right) \in \FF_q^{n+n'+h}$,
the condition 
\begin{align}
z_{j}=\sum_{k=1}^{n+n'+h}m_{\varsigma}\left(j,k\right) x_{k}
\Label{H201}
\end{align}
can be rewritten as
\begin{align}
\vec{z} =M_{\varsigma,1} \vec{a} +M_{\varsigma,2}\vec{b} +M_{\varsigma,3}\vec{c}.
\Label{H201B}
\end{align}
\begin{lemma}\Label{L88}
Secrecy holds for Eve's attack on $E_A$ if and only if the following condition holds:
For any vector $\ba \in \FF_q^n$, there exists a function $\vec {\mathfrak b}(\ba) \in \FF_q^{n'}$ such that
\begin{align}
M_{\varsigma,1} \vec{a} = M_{\varsigma,2}\vec {\mathfrak b}(\ba). 
\label{eq:g_secure_1}
\end{align} 
\end{lemma}
The condition is trivially equivalent to the condition that the image of $M_{\varsigma,1}$ is contained in that of $M_{\varsigma,2}$.

\begin{proof}
Due to Lemma \ref{L2-4}, it is enough to discuss the case with $\vec{C}=\vec 0$.
When secrecy holds,
$
\{ M_{\varsigma,2}\bb | \bb \in \FF_q^{n'}\}
=
\{M_{\varsigma,1} \vec{a}+ M_{\varsigma,2}\bb | \bb \in \FF_q^{n'}\}
$ for any $\ba \in \FF_q^n$.
The latter set contains 
$M_{\varsigma,1}\ba$, which ensures the existence of $\vec {\mathfrak b}(\ba)$.
When such a function $\vec {\mathfrak b}(\ba)$ exists,
the distribution of $M_{\varsigma,2}\vec{B} $
is the same as that of 
$M_{\varsigma,1} \vec{a}+ M_{\varsigma,2}\vec{B}$.
That is because  the variable $\vec{B}$ is  uniformly  distributed.
This fact implies the secrecy.
\end{proof}

\subsection{Recoverability against Eve's attack\Label{sub:Robustness-against-Eve's}}
For our analysis of
 deriving
 quantum network coding, 
we need  to
introduce the
concept of
recoverability of 
 the 
 classical network code against Eve's attack 
in addition to the secrecy.
The concept of recoverability 
is defined as follows.
We consider the situation that Eve can  disturb contents on the channels in $E_A$ as is done in the case of secrecy analysis.
In other words, she can inject any contents  on the channels in $E_A$.
In such a situation,
we imagine a
  receiver 
Bob
  who can use all the received contents of the the channel identified by a set $E_P\subset E\backslash E_R$.
For convenience,
 we give a name ``protected edges'' to the edges in $E_P$.
  He can additionally access all the shared-random variables,
and can know the set $E_{A}$ and the network structure, i.e. the matrix $M'$.
However, Bob does not know what content is injected in the channel in  the set $E_A$, if the channel is not in the set $E_P$.
In this case, if Bob can reconstruct the original messages, we call that the messages is recoverable from Eve's attack by the protected edges $E_P$.
We will
require 
 the recoverability by a certain subset $E_P$ 
for the security of a deriving quantum network code.

\begin{table}[htpb]
  \caption{Sets of edges  of the network coding.}
\Label{T1-2}
\begin{center}
  \begin{tabular}{|c|l|} 
\hline
$\tilde E$ & The set of edges which express actual channels 
\\
\hline
$E_I$ & 
The set of input edges connected from input vertices 
\\
\hline
$E_O$ & 
The set of output edges connected to output  vertices 
\\
\hline
$E_R$ &
The set of shared-randomness edges connected from  shared-randomness vertices 
\\
\hline
$ E$ & The union of the sets $\tilde E$, $E_I$, $E_O$, and $E_R$
\\
\hline
$ E_P$ & The set of protected edges\\
\hline
$E_A$ & The set of edges attacked by Eve
\\
\hline
  \end{tabular}
\end{center}
\end{table}

Now, we give more rigid definition of the concept of the recoverability.
For the subset $\Ed$, we define
the strictly increasing function $\iota:\ \left\{ 1,\cdots,\left| \Ed \right|\right\} \rightarrow\left\{ 1,\cdots,n,n+l+1,\cdots,N+2n+l\right\} $
which
satisfies $\Ed=\left\{ \mathbf{e}\left(\iota\left(j\right)\right)\right\} _{j=1}^{h'}$, where $h':=\left| \Ed\right|$.
Then, the contents  
$\left\{ Y'_{\iota\left(j\right)}\right\} _{j=1}^{h'}$
received from the protected edges $\Ed$
 can be written as 
\begin{align}
Y'_{\iota\left(j\right)}= 
\sum_{k=1}^{n+n'+h}m_{\iota}'\left(j,k\right)X_{k}, \Label{e14}
\end{align}
where $m_{\iota}'\left(j,k\right):=m'\left(\iota\left(j\right),k\right)$ is a matrix elements of the $h'\times (n+n'+h)$ matrix  $M'_\iota$.
Then, we 
 rigidly
define the concept of the recoverability as follows.

\begin{definition}\Label{DD5}
We call that the messages are recoverable for Eve's attack on $E_A$ by a subset $\Ed$, when 
for any vector $\vec{b}\in \FF_{q}^{n'}$,
there exists a function $f_{\vec{b}}:\ \FF_{q}^{h'}\rightarrow
\FF_{q}^{n}$ such that
\begin{equation}
f_{\vec{b}}\left(
M_{\iota}'\cdot
(\vec{a}, \vec{b}, \vec{c})^T
\right)
=\vec{a}\Label{eq:sec classical subsec robustness f a M i a e =00003D i}
\end{equation}
for any $\vec{a} \in \FF_q^n$ and $\vec{c}\in \FF_q^h$.
\end{definition}
Here, $T$ means the transposition.
Note that the matrix $M_{\iota}'$ is uniquely given only from the matrix $M_0$ and the sets $E_A$ and $E_P$.

The function $f_{\vec{b}}$ is nothing but a decoder of the messages $\vec{A}$ from
the contents received from  $\Ed$.
The function depends on 
 $M_\iota'$ and $\vec b$
 only. 
Since condition 
\eqref{eq:sec classical subsec robustness f a M i a e =00003D i}
does not depend on the choice of $\vec{c}$,
it guarantees the recoverability even when 
Eve chooses $\vec{C}$ depending on her wiretapped variable.  

Notice that this kind of recoverability does not imply the recoverability of the messages by 
terminal nodes, and
 we don't assume the condition $ \Ed \cap E_A=\emptyset$.
 In other words,  a channel corresponding to a ``protected'' edge may be disturbed
by Eve.
Therefore, at the channel  in $ \Ed \cap E_A$,
 Eve can completely control the information obtained by Bob.
 
It is informative to show a toy example of a classical network code in which the contents from the edges in $ \Ed \cap E_A$ are useful to recover the messages.
The example is as follows. A single message $A_1$ is transfer from $s_1$ to $t_1$ via two channels $\mathbf e(2)$ and $\mathbf e(3)$  
simultaneously. There is no randomness. The terminal node $t_1$ sums up the  two received contents and obtains recovered message by dividing it by $2$.  $\mathbf e(1)$ and $\mathbf e(4)$ are the input edge and output edge respectively.
We define the set $E_P$ to be $\{\mathbf e(3),\mathbf e(4)\}$, and consider the case $E_A=\{\mathbf e(3)\}$.
In this case, we can give 
\begin{equation}
M'_{\iota}=
\left(
\begin{array}{cc}
0&1\\
2^{-1}&2^{-1}
\end{array}
\right),
\end{equation}
and we know that   $\vec{A} \in \FF_q^1$ and $\vec{C}\in \FF_q^1$.
Therefore, by selecting the function $f_{\vec b}(\vec y)$ as $2y_2-y_1$, we can check that the message is recoverable for Eve's attack  on  $E_A=\{\mathbf e(3)\}$, though this network code isn't secure against Eve's attack on $E_A$. 
The necessity of the content from the channel   $\mathbf e(3)\in E_P\cap E_A$ is checked from the fact that 
 the  coefficient of  $y_1$ for the function  $f_{\vec b}(\vec y)$ is not 0.

In the end of this section, the defined matrices in this section are summarized in Table \ref{T2}. 
\begin{table}[htpb]
  \caption{Summary of matrices}
\Label{T2}
\begin{center}
  \begin{tabular}{|c|l|l|c|} 
\hline
matrix & input system & output system & equation  \\
\hline
\multirow{2}{*}{$M_0$} & \multirow{2}{*}{messages, shared random variables} & 
inputs of all edges  & \multirow{2}{*}{\eqref{H204}}\\
&  & $=$ outputs of all edges  & \\
\hline
\multirow{2}{*}{$M$}& messages, shared  random variables,  & 
\multirow{2}{*}{inputs of all edges} &\multirow{2}{*}{\eqref{eq:sec classical subsec security y i =00003D M i k x k}}
\\
& Eve's input &  &\\
\hline
\multirow{2}{*}{$M'$} & messages, shared random variables, & 
\multirow{2}{*}{outputs of all edges} &\multirow{2}{*}{\eqref{eq:M' ik =00003D delta M ik}}
\\
&Eve's input & &\\
\hline
\multirow{2}{*}{$M_\varsigma$} & messages, shared random variables,  & 
\multirow{2}{*}{inputs of attacked edges}&\multirow{2}{*}{\eqref{H201}}
\\
&Eve's input& &
\\
\hline
\multirow{2}{*}{$M_\iota'$} & messages, shared random variables,  
& \multirow{2}{*}{outputs of protected edges}
& 
\multirow{2}{*}{\eqref{e14}}
\\
&Eve's input &  &\\
\hline
  \end{tabular}
\end{center}
\end{table}

\section{Secure quantum network coding for general network}\Label{S6}
\subsection{Coding scheme}
In this section, we 
derive a  quantum network code from a linear classical network code, and 
analyze the security of the 
 quantum network coding based on 
 the 
  properties 
of
 the original 
 classical network coding 
  which are discussed 
 in the previous section. 
Quantum network coding can be categorized by the type of classical communication allowed
\cite{Hayashi2007,PhysRevA.76.040301,Kobayashi2009,Leung2010,Kobayashi2010,Kobayashi2011}.
In this paper, we consider the case that
 any authenticated public 
classical communication
 from any nodes 
to the terminal nodes 
is freely available, and all the communication may be eavesdropped by Eve. 
In this case, it is known that, 
 for an arbitrary classical multiple-unicast code on an arbitrary classical network, 
there exists a corresponding quantum multiple-unicast network
code on the corresponding quantum network \cite{Kobayashi2010,Kobayashi2011}.
We start this subsection by extending this known result to the case when shared randomness is employed.

\subsubsection{ The notations defined from the original classical network coding}
 
In the following, we fixed the original classical network code. 
As is done in the previous section, from the classical network code, we define integers
$N$, $n$, $n'$, $l$, $q$, $\{l_k\}_{k=1}^{n'}$,
 sets 
 $\tilde{V}$, $V_S$, $V_T$, $E$, $\tilde{E}$, $E_I$, $E_O$, $E_R$, $E_P$, maps 
$\mathbf{e}$, $\mathbf{v}_{I}$, $\mathbf{v}_{O}$,  $\inc$, and  the coefficients 
$\left\{ \theta_{j,k}\right\} _{j\in\left\{n+l+1,\cdots,|E|\right\},k \in \inc\left(j\right)}$,
 which identify the matrix $M_0$ and its elements $m_0(j,k)$ by Eq.\eqref{H204}.

Other than the above notations, we have to define additional notation of a map $\mathfrak E$
from the element in 
$  E$ to the subset of $V_T$ such that
\begin{align}
\mathfrak E(\mathbf{e}'):=&\{\mathbf{v}_{I}(\mathbf{e}(N+n+l+j))|1\leq j\leq n\land  \exists k, \;(\mathbf{e}'=\mathbf{e}(k)\land m_0(k,j)\neq 0 )\}.
\end{align}
Note that, 
when the content tranferred by the edge $\mathbf{e}'$
depends on some messages, the set $\mathfrak E(\mathbf{e}')$ indicates that of 
 all the terminal nodes where the messages are reconstructed.

\subsubsection{
Considering situation of the quantum network}
 
From the items defined above, we list the conditions of the considering situation as a quantum network:

\begin{itemize}
\item
The number of nodes of the quantum network coding is $N$, and each member is  labelled by an element of $\tilde V$ individually.
\item
The total Hilbert space, which all the nodes treat, is the direct product of the subspaces $\mathcal H_j$
for $1\leq  j\leq n$ or $ n+l< j\leq N+2n+l$.
Every subspace $\mathcal H_j$ is made from a  $q$-dimensional Hilbert space, and  has a computational
basis $\left\{ \ket{k}_j\right\} _{k \in \FF_q}$. 
Every subspace $\mathcal H_j$ of the first $n$ subspaces  is occupied by the node $\mathbf{v}_{O}\left(\mathbf{e}\left(j\right)\right)$ for each $j$.
Every subspace $\mathcal H_j$ of the other  subspaces  is occupied by the node $\mathbf{v}_{I}\left(\mathbf{e}\left(j\right)\right)$  for each $j$.
\item
Initially, there is no correlation, especially no entanglement, between any pair of nodes except for the preshared quantum messages.
\item
At the time $t=j$, 
we can use  a quantum channel identified by $\mathbf{e}\left(j\right)$ 
which transfers the quantum subspace $\mathcal H_j$ from the node
$\mathbf{v}_{I}\left(\mathbf{e}\left(j\right)\right)$ 
to the node $\mathbf{v}_{O}\left(\mathbf{e}\left(j\right)\right)$ where  $n+l< j\leq N+n+l$.
Any channel can be used  only once, and
any channel is an identity channel if the eavesdropper Eve does not attack the channel.
\item
A random number $b_j\in\mathbb F_q$, which is secret from Eve, is sheared by the nodes 
 $\mathbf{v}_{O}\left(\mathbf{e}\left(k\right)\right)$ for $ n+\sum_{j'=1}^{j-1}l_{j'}<k \leq n+\sum_{j'=1}^jl_{j'}$ initially where $1\leq j\leq n'$. Other than the random numbers, the vertex $\mathbf{v}_{O}\left(\mathbf{e}\right)$ shares a secret random number in $\mathbb F_q$ with all vertices in $\mathfrak E(\mathbf{e})$ for every $\mathbf{e}\in E_P$.
\item
Any node can apply any unitary operations and measurements for the occupied quantum subspaces depending on any classical information which the node has at any time.
\item
Any authenticated but public 
 classical communication is freely available from any node to all of the terminal nodes.
That is, each node can freely send classical information to any terminal node, and the information may be revealed to Eve.
\end{itemize}

\subsubsection{ Purpose of the  quantum network coding}

There are two purposes for the multiple-unicast quantum network code.
The first purpose is to send
an arbitrary quantum state on $\mathbb{C}^{q}$ from a source node
$\mathbf{v}_{O}\left(\mathbf{e}\left(j\right)\right)$ 
to a terminal node $\mathbf{v}_{I}\left(\mathbf{e}\left(N+n+l+j\right)\right)$
for all $j\in\{1,2,\cdots, n\}$ through the quantum network simultaneously. 
We call the state a quantum message.
Since any classical communication to terminal nodes is free,
this task is equivalent to constructing the maximally
entangled state  between
a $q$-dimensional subspace in
a source node 
$\mathbf{v}_{O}\left(\mathbf{e}\left(j\right)\right)$ 
 and 
that in
 a terminal node
 $\mathbf{v}_{I}\left(\mathbf{e}\left(N+n+l+j\right)\right)$ for all $j\in\{1,2,\cdots, n\}$.
Second purpose is to prevent the leakage of  any information about the quantum messages to Eve where
she can access all the information transmitted via public classical channel and quantum states as contents on the restricted quantum channels identified by $E_A$. 

 In this paper,  we will show some examples of quantum network codes which satisfies the following two properties.
First, quantum messages can be sent with fidelity 1, if there is no disturbance for any channels.
Second, even if any one or two edges are completely controlled by Eve, i.e. the transmitted contents are completely stolen and other contents are injected on any one or two edges in $\tilde E$, it can be guaranteed that Eve can get no information about the quantum messages.

\subsubsection{Preliminary definition of the quantum network coding}

Before presenting the quantum network code, we give the notations used in it. 
For a subset $D$ of $\{1,\cdots,N+2n+l\}$, we define the subspace 
${\mathcal{H}}_{D}:=\bigotimes_{j\in D}{\mathcal{H}}_{j}$.
For  an $\FF_q$-valued vector $\vec{y}=\left(y_{1},\cdots,y_{N+2n+l}\right) \in \FF_q^{N+2n+l}$,
we abbreviate the state $\bigotimes_{j\in D}\ket{y_j}_{j}$ as  $\ket{\vec{y}}_{D}$. Note that, from this definition,  a single vector has multiple expressions in order to simplify the expressions hereafter.    
To distinguish a classical system from a quantum one easily, we introduce sets
\begin{align*}
\qin\left(j\right) & :=\left\{ k\in \inc\left(j\right)| 1\leq k\leq n\;\lor \;n+l< k\right\} \\
\cin\left(j\right) & :=\left\{ k\in \inc\left(j\right)|n< k\le n+l\right\} ,
\end{align*}
where $\inc\left(j\right)$ is defined by Eq.\eqref{eq:def in j}. 
Using these notations, 
depending on the matrix $\theta= \left\{ \theta_{j,k}\right\}_{j,k}$,
we define the controlled unitary operation
$U_{j }(\theta)$
acting on the Hilbert space 
${\mathcal{H}}_{j}\otimes{\mathcal{H}}_{\qin\left(j\right)}$ as 
\begin{align*}
U_{j}(\theta) 
&=&
q^{-N-2n-l+1+|\qin\left(j\right)|}
\sum_{\vec y\in\mathbb{F}_{q}^{N+2n+l}}
| y_{j} +\sum_{k\in \qin\left(j\right)} \theta_{j,k} y_k\rangle_j
\bra{y_j}_j
\otimes\ket{\vec y}_{\qin(j)}\bra{\vec y}_{\qin(j)} .
\end{align*}
On the space ${\mathcal{H}}_{j}$, whose computational basis is $\{ \ket{y}_{j}\} _{y \in \FF_q}$,
 we introduce  the Fourier basis 
$\{ |\tilde \beta\rangle_{j}\} _{\beta \in \FF_q}$
as 
\[
|\tilde \beta\rangle_{j}:=q^{-1/2}\sum_{y \in \FF_q} \omega^{\tr y\beta }\ket{y}_{j},
\]
where $\omega:=\exp\left(-\frac{2\pi i}{p}\right)$.
Here, $\tr z $ expresses 
the element $\Tr \psi(z) \in \FF_p$, 
where $\psi(z)$ denotes the matrix representation of the multiplication map $x \mapsto zx $ 
which  identifies the finite field $\FF_{q}$ with the vector space $\FF_p^d$, where $d$ is the degree of algebraic extension of $\FF_{q}$, i.e. $p^d=q$.
For the details, see \cite[Section 8.1.2]{Haya2}.
We also define the generalized Pauli operators $\sX_j(x)$ and $\sZ_j(\beta)$ 
as $\sX_j(x):=\sum_{y \in \FF_q}\ket{y+x}_j\bra{y}_j$
and $\sZ_j(\beta):=\sum_{y \in \FF_q}\omega^{\tr y\beta}\ket{y}_j\bra{y}_j$.

\subsubsection{Quantum network code}
Using the notations defined above, we show the multi-unicast quantum network code which transfers the quantum messages  from the space $\bigotimes_{j=1}^n\mathcal H_{j} $ into the space $\bigotimes_{j=1}^n\mathcal H_{N+n+l+j} $. 

\begin{Protocol}[h]
\caption{The quantum network code
 deriving from a general classical linear network code}         
\Label{protocol2}
\begin{algorithmic}
\STEPONE {\bf Initialization}

First, all the spaces $\mathcal H_j$ are initialized to the state $\ket{0}_j$
for $n+l<j\leq N+2n+l$, at each edge.

\STEPTWO {\bf Transmission}

This step consists of $N+n$ substeps.
The $j$-th substeps 
 can be described as follows.
At the time $t=j':=n+l+j$,
 the node $\mathbf{v}_{I}\left(\mathbf{e}\left(j'\right)\right)$
operates the unitary 
\begin{align}
\sX_{j'}( \sum_{k \in \cin\left(j'\right)}
\theta_{j',n+k} b_k )
U_{j'}\left(
\theta
\right)
\label{H3-14}
\end{align}
on
 ${\mathcal{H}}_{j'}\otimes{ \mathcal{H}}_{\qin\left(j'\right)}$ where
 ${\mathcal{H}}_{j'}$ is the controlled system and ${\mathcal{H}}_{\qin\left(j'\right)}$ is the controlling system.
 If $j\leq N$, the
 node
$\mathbf{v}_{I}\left(\mathbf{e}\left(j'\right)\right)$ sends the Hilbert
space ${\mathcal{H}}_{j'}$ to the node $\mathbf{v}_{O}\left(\mathbf{e}\left(j'\right)\right)$
via the quantum channel $\mathbf{e}\left(j'\right)$. 

  Note that, if  the node $\mathbf{v}_{I}\left(\mathbf{e}\left(j'\right)\right)$
 does not share any random number, i.e. $\cin\left(j'\right)=\emptyset$,  the generalized Pauli operator $\sX_{j'}(\cdot)$ in the above relation is considered to be the identity operator.

\STEPTHREE {\bf Measurement on Fourier-basis}

 This step consists of $N+n$ substeps.
The step identified by $j\in\{1,\cdots, n,n+l+1,\cdots N+n+l\}=:G'$
 can be described as follows.
 The node $\mathbf{v}_{O}\left(\mathbf{e}\left(j\right)\right)$
measures the Hilbert space ${\mathcal{H}}_{j}$ in the Fourier basis,
and sends the measurement outcome $\beta_{j}$ to all the terminal nodes in
$\mathfrak E(\mathbf{e}\left(j\right))$. 
Here,
if  $\mathbf{e}\left(j\right) \notin E_P$, 
the outcome is sent by public channel, i.e. the outcome may eavesdropped by Eve, 
and, 
if  $\mathbf{e}\left(j\right) \in E_P$, 
the outcome is sent by the one-time pad, i.e. a secret randomness shared with the vertices in $\mathfrak E(\mathbf{e}\left(j\right))$ is consumed and
the outcome is completely secret from Eve.

\STEPFOUR {\bf Recovery}

For all $j$ satisfying $1\le j\le n$, the terminal node $\mathbf{v}_{I}\left(\mathbf{e}\left(N+n+l+j\right)\right)$
operates 
$\sZ_{N+n+l+j} \left(
\sum_{k\in G'}\beta_{k}m_0\left(k,j\right)
\right)$, where a matrix $M_0$ is defined by Eq.\eqref{H204}. 
\end{algorithmic}
\end{Protocol}
Note that, for all public communications sending  an outcome $\beta_j$ to multiple nodes at a substep in Step 3, we can combine a common single secret randomness for the one-time pad without losing secrecy. Furthermore,  
there is a special case such that  $\mathfrak E( \mathbf{e}(j)  )$  contains only the single node $\mathbf{v}_{O}(\mathbf{e}(j))$.  In that case,  we send 
 the outcome to the node $\mathbf{v}_{O}(\mathbf{e}(j))$ where the outcome is obtained. Therefore, the procedure is equivalent to doing nothing. As a result, we don't have to use any shared randomness even if  $\mathbf{e}(j) \in E_P$  for such a situation.

As you have seen, our protocol depends only on  the set of coefficients $\{\theta_{j,k}\}_{j,k}$ and the set of protected edges $E_{P}$.
That is, our protocol is uniquely determined by the pair of $\{\theta_{j,k}\}_{j,k}$ and $E_{P}$,
and we call it the quantum network code 
$\left\{ \theta_{jk}\right\} _{j\in\left\{n+l+1,\cdots,|E|\right\},k \in \inc\left(j\right)}$
with the set of protected edges $E_{P}$.

\subsection{Validity analysis}

In order to analyze the quantum network coding,
it is convenient to  introduce ancillary set of $q$-dimensional Hilbert spaces ${\mathcal{H}}_{j-n}$ occupied by the source node $\mathbf{v}_{O}\left(\mathbf{e}\left(j\right)\right)$  for $1\le j\le n$. 
Note that we never perform any operations on the ancillary spaces.

As a generalization of \cite[Theorem 1]{Kobayashi2010},
we obtain the following theorem.

\begin{theorem}\Label{T2-5-1}
Suppose that the corresponding classical network coding identified by 
$\{\theta_{j,k}\}_{j,k}$ is a multi-unicast network code.
By Protocol \ref{protocol2}, any quantum message on the space $\mathcal{H}_{j}$
are simultaneously transferred to the space  $\mathcal{H}_{N+n+l+j}$ with fidelity $1$ 
for any $j$ satisfying $1\leq j\leq n$
if no one disturbs the protocol.
That is, 
when the maximally entangled state $q^{-1/2}\sum_{x\in\mathbb F_q}\ket{x}_{j-n}\ket{x}_{j}\in{{\mathcal{H}}_{j-n}\otimes \mathcal{H}}_{j}$ 
is prepared as the initial state on every source node $\mathbf{v}_{O}\left(\mathbf{e}\left(j\right)\right)$ for $1\le j\le n$,  Protocol \ref{protocol2} makes the resultant state to be a maximally entangled state $q^{-1/2}\sum_{x\in\mathbb F_q}\ket{x}_{j-n}\ket{x}_{N+n+l+j}$  on
${\mathcal{H}}_{j-n}\otimes{\mathcal{H}}_{N+n+l+j}$ for any $j$ satisfying
$1\le j\le n$ if all the quantum channels are identity channels. 
\end{theorem}
Remember that the transmission of quantum states is mathematically equivalent to 
sharing the maximally entangled state between the input and output systems.

\begin{proof}
We define the Hilbert spaces ${\mathcal{H}}_{I}$ and ${\mathcal{H}}_{O}$,
 as ${\mathcal{H}}_{I}:=\bigotimes_{j=1}^{n}{\mathcal{H}}_{j-n}$
and 
${\mathcal{H}}_{O}:=\bigotimes_{j=1}^{n}{\mathcal{H}}_{N+n+l+j}$ respectively.
Their bases $\{\bigotimes_{j=1}^n\ket{a_j}_{j-n}\}$ and $\{\bigotimes_{j=1}^n\ket{a_j}_{N+n+l+j}\}$ are abbreviated as 
$\{\ket{\vec a}_{I}\}$ and $\{\ket{\vec a}_{O}\}$. The sets $G$ and $G'$ are defined to be $\{1,\cdots ,n,n+l+1,\cdots,N+2n+l\}$ and 
$\{1,\cdots ,n,n+l+1,\cdots,N+n+l\}$.
   By straightforward calculation, we find that the density matrix on the network
after Step 2 is  
\begin{align*}
&\frac{1}{q^{n+n'}} \sum_{\vec{b}\in\mathbb{F}_{q}^{n'}}
\sum_{\vec{a},\vec{a}'\in\mathbb{F}_{q}^{n}}
\ket{\vec{a}}_{I}\bra{\vec{a}'}_{I}\otimes\ket{M_{0}\cdot
(\vec{a},\vec{b})^T
}_{G}\bra{M_{0}\cdot
(\vec{a}',\vec{b})^T
}_{G}\\
= & \frac{1}{q^{n+n'}}\sum_{\vec{b}\in\mathbb{F}_{q}^{n'}}
\sum_{\vec{a},\vec{a}'\in\mathbb{F}_{q}^{n}}
\ket{\vec{a}}_{I}\bra{\vec{a}'}_{I}\otimes
\ket{M_{0}\cdot
(\vec{a},\vec{b})^T
}_{G'}\bra{M_{0}\cdot
(\vec{a}',\vec{b})^T
}_{G'}
\otimes\ket{\vec{a}}_{O}\bra{\vec{a}'}_{O},
\end{align*}
if all the quantum channels are identity channels.
At the equality, we use the assumption that the classical protocol is a multiple-unicast network code.
The state after Step 3 can be expressed as
\[
\frac{1}{q^{n}}
\sum_{\vec{a},\vec{a}'\in\mathbb{F}_{q}^{n}}
\omega^{ \tr \vec{\beta}_{G'}^{T}\cdot M_{0}\cdot\left(
\vec{a}'-\vec{a},
\vec{0}_{n'}
\right)^T}\ket{\vec{a}}_{I}\bra{\vec{a}'}_{I}\otimes
\ket{\vec{a}}_{O}\bra{\vec{a}'}_{O},
\]
where $\vec{\beta}_{G'}:=(\beta_{1},\cdots,\beta_{n},\vec 0_l,\beta_{n+l+1},\cdots,\beta_{N+n+l},\vec 0_n)^T$.
 Finally, the state after Step 4 can be written as 
\[
\frac{1}{q^{n}}\sum_{\vec{a},\vec{a}'\in\mathbb{F}_{q}^{n}}
\ket{\vec{a}}_{I}\bra{\vec{a}'}_{I}\otimes
\ket{\vec{a}}_{O}\bra{\vec{a}'}_{O},
\]
which is the maximally entangled state to be constructed in this protocol. 
\end{proof}

\subsection{Security analysis}\Label{S6-2}
Next, we discuss the security of the transmitted quantum state under the following four assumptions.
First, the eavesdropper Eve can eavesdrop and modify the contents transmitted via all the channels  in $E_A$, which is a subset of $\tilde E$.
Second, she also knows the network structure, i.e., the topology of the network and all the coefficients $\{\theta_{j,k}\}_{j,k}$.
Third, Eve can get any information transmitted by the public channel.
Finally, Eve can't obtain any other information which may be correlated to the quantum messages.

In order to treat Eve's attack formally,
we introduce the map $\varsigma$ and the constant $h$ defined from $\mathbf{e}$ and $E_A$
as is done in the case of Eve's attack for the classical network coding.
Using this notation, we formulate the Eve's attack as follow.

{\bf Eve's attack:} 
Eve initially occupies  her initial Hilbert space $\mathcal{W}$ with a state $|\phi_{ini}\rangle$, 
where the dimension of the space $\mathcal W$ is chosen to be sufficiently large so that every Eve's operations can be treated as a unitary operation.
At the time $t=\varsigma\left(j\right)$,
Eve applies the unitary $W_{j}$ on 
$\cH_{\varsigma\left(j\right)} \otimes \cW$ for $1\leq j\leq h$.
Note that $W_{j}$ does not depend on the outcomes $\{\beta_k\}_k$ since 
the measurement step is done just after the transmission step.
However, Eve may finally get the measurement outcomes $\beta_k$ where $1\leq k\leq n$ or $n+l<k\leq N+n+l$ and $\forall j,\;
k\neq \iota(j)$, i.e. the measurement outcomes of the contents received from non-protected edge.
In the following security analysis, these classical information is denoted by a diagonalized density matrix on the space $\mathcal{V}$,
where the initial state of $\mathcal{V}$ is a pure state.

From this assumption, we also formulate the security of the quantum network coding against Eve's attack:
\begin{definition}
The quantum network code 
$\left\{ \theta_{j,k}\right\} _{j\in\left\{n+l+1,\cdots,|E|\right\},k \in \inc\left(j\right)}$
with the set of protected edges $E_{P}$ is called secure for Eve's attack 
$\{ V_j\}_j$ on the set of edges $E_{A}$
if the following condition holds.
When the initial state on the Hilbert space $\mathcal{H}_{I}\otimes \bigotimes_{j=1}^n \mathcal H_{j}$
is the maximally entangled state between $\mathcal{H}_{I}$ and $\bigotimes_{j=1}^n \mathcal H_{j}$
 i.e. the initial state is that  used in  Theorem \ref{T2-5-1},
the final state of the protocol on the subspace $\mathcal{H}_{I}\otimes\mathcal{W} \otimes\mathcal{V}$ is 
 a product state with respect to the partition between $\mathcal{H}_{I}$ and $\mathcal{W} \otimes\mathcal{V}$. 
\end{definition}
It is easily understood that: this defined condition   of the security 
is equivalent to the condition that
there is no leakage of the information  about the quantum messages by the quantum network code.
Note that, we call the state $\rho\in\mathfrak{B}\left({\mathcal{H}}'\otimes{\mathcal{H}}''\right)$
 a product state if there exist $\rho'\in\mathfrak{B}\left({\mathcal{H}}'\right)$
and $\rho''\in\mathfrak{B}\left({\mathcal{H}}''\right)$ such that $\rho=\rho'\otimes\rho''$.
Now, we can present the main result of this paper: 

\begin{theorem}
\Label{thm:The-quantum-network main}
The quantum network code 
$\left\{ \theta_{j,k}\right\} _{j\in\left\{n+l+1,\cdots,|E|\right\},k \in \inc\left(j\right)}$
with the set of protected edges $E_{P}$ is secure for all Eve's attacks on the set of edges $E_{A}$ 
if the following two conditions hold.
(i) The classical network code 
$\left\{ \theta_{j,k}\right\} _{j\in\left\{n+l+1,\cdots,|E|\right\},k \in \inc\left(j\right)}$
is secure for Eve's attacks on the set of edges $E_{A}$.
(ii) The messages are recoverable for Eve's attack on   $E_{A}$ 
by the set of protected edges $E_{P}$
in the sense of the classical network coding. 
\end{theorem}
From this theorem, we know that the security for the quantum messages is related not only to the secrecy of the classical information but also to the recoverability of the classical information.
Strictly speaking, this theorem guarantees that the security analysis of our quantum network coding 
is reduced to the analysis of the secrecy and the recoverability of the corresponding classical network coding.

\subsection{Security proof}\label{ProofTh}
We can prove Theorem 2 by checking Definition 2 directly as follows:

\emph{Proof of Theorem }\ref{thm:The-quantum-network main}: 

We consider the case that 
we initialize the state on the Hilbert space $\mathcal{H}_{I}\otimes \bigotimes_{j=1}^n \mathcal H_{j}$
to be  the maximally entangled state between $\mathcal{H}_{I}$ and $\bigotimes_{j=1}^n \mathcal H_{j}$, i.e. $\bigotimes_{j=1}^nq^{-1/2}\sum_{a\in\mathbb F_q}\ket{a}_{j-n}\ket{a}_{j}$,
and execute Protocol 1.

Given
 Eve's attack $\{V_j\}$ on the set of edges $E_A$, the total density matrix $\rho$
on the space $\mathcal H_I\otimes \mathcal H_O\otimes \mathcal V'\otimes \mathcal W\otimes \mathcal V$
 becomes
\begin{align}
&q^{-2n-n'-l}\sum_{\vec a,\vec a'\in\mathbb F_q^{n}}
\sum_{\vec b\in\mathbb F_q^{n'}}
\sum_{\vec c,\vec c'\in\mathbb F_q^{h}}
\sum_{\vec \beta\in\mathbb F_q^{N+2n+l}}
\ket{\vec a}_I\bra{\vec a'}_I
\nonumber\\
&
\otimes \langle\tilde {\vec  \beta}|_{G'}|M'(\vec a,\vec b,\vec c)^T\rangle_G 
\langle M'(\vec a',\vec b,\vec c')^T|_G|\tilde {\vec  \beta}\rangle_{G'}
\otimes \bigotimes_{j\in G'}\ket{\beta_j}_j^{\mathcal V'} \bra{\beta_j}_j^{\mathcal V'}
\nonumber\\
&
 \otimes  
(\prod_{j=1}^h\bra{ c_j}_{\varsigma (j)}V_j|M(\vec a,\vec b,\vec c)^T\rangle_{\varsigma (j)})\ket{\phi_{ini}}
\bra{\phi_{ini}}
(\prod_{j=1}^h\bra{ c_j'}_{\varsigma (j)}V_j|M(\vec a',\vec b,\vec c')^T\rangle_{\varsigma (j)})^\dagger
\nonumber\\
&
\otimes \bigotimes_{j\in G'\backslash H}\ket{\beta_j}_j^{\mathcal V}\bra{\beta_j}_j^{\mathcal V},
\end{align}
after Step 3,
where all the outcomes shared by  terminal nodes are denoted by a diagonal density matrix on the space $\mathcal V'$.
Note that the bases of $\mathcal V$ and $\mathcal V'$ are expressed by
 $\{\bigotimes _j\ket{\beta_j}_j^{\mathcal V}\}$
and 
 $\{\bigotimes _j\ket{\beta_j}_j^{\mathcal V'}\}$ respectively, and
we abbreviate the state 
$|\tilde \beta_{j_1}\rangle_{j_1}\otimes\cdots \otimes|\tilde \beta_{j_m}\rangle_{j_m}$ as 
$|\tilde {\vec\beta}\rangle_{(j_1,\cdots,j_m)}$  where $\vec \beta=(\beta_1,\beta_2,\cdots)$
as is the case of the computational base, and $H$ is defined to be the set  
$\{j|\mathbf e(j)\in E_P\}$. 
Since  all the operators in Step 4 of Protocol 1 are operators closed in the space $\mathcal H_O\otimes\mathcal V'$,
 it is sufficient  to check that the partial trace of $\rho$ with respect $\mathcal H_O\otimes\mathcal V'$, which is equal to   
\begin{align}
&q^{-2n-n'-l}\sum_{\vec a,\vec a'\in\mathbb F_q^{n}}
\sum_{\vec b\in\mathbb F_q^{n'}}
\sum_{\vec c,\vec c'\in\mathbb F_q^{h}}
\sum_{\vec \beta\in\mathbb F_q^{N+2n+l}}
\ket{\vec a}_I\bra{\vec a'}_I
\nonumber\\
&
\times {\rm Tr}(\langle\tilde {\vec  \beta}|_{G'}|M'(\vec a,\vec b,\vec c)^T\rangle_G 
\langle M'(\vec a',\vec b,\vec c')^T|_G|\tilde {\vec  \beta}\rangle_{G'})
\nonumber\\
&
 \otimes  
(\prod_{j=1}^h\bra{ c_j}_{\varsigma (j)}V_j|M(\vec a,\vec b,\vec c)^T\rangle_{\varsigma (j)})\ket{\phi_{ini}}
\bra{\phi_{ini}}
(\prod_{j=1}^h\bra{ c_j'}_{\varsigma (j)}V_j|M(\vec a',\vec b,\vec c')^T\rangle_{\varsigma (j)})^\dagger
\nonumber\\
&
\otimes \bigotimes_{j\in G'\backslash H}\ket{\beta_j}_j^{\mathcal V}\bra{\beta_j}_j^{\mathcal V},
\label{eq:g_1}
\end{align}
is a product state with respect to the partition between $\mathcal H_I$ and $\mathcal W\otimes\mathcal V $.
To simplify this expression, we use the following relation: for any density matrix $\rho_G\in \mathcal H_G$, any function $g$,  and any sets $D, D'$ which satisfies $ D,D'\subset G$,
\begin{align}
&
\sum_{\vec \beta\in\mathbb F_q^{N+2n+l}}
g(\{\beta_j\}_{j\in D'\cap D}) {\rm Tr}(\langle\tilde {\vec  \beta}|_{D'}\rho_G|\tilde {\vec  \beta}\rangle_{D'})
\nonumber\\
=&
\sum_{\vec \beta\in\mathbb F_q^{N+2n+l}}
q^{|G|-|D'|}g(\{\beta_j\}_{j\in D'\cap D})
\langle\tilde {\vec  \beta}|_{G}\rho_G|\tilde {\vec  \beta}\rangle_{G}
\nonumber\\
=&
\sum_{\vec \beta\in\mathbb F_q^{N+2n+l}}
q^{|G|-|D'|}g(\{\beta_j\}_{j\in D'\cap D})
\langle\tilde {\vec  \beta}|_{ D}\otimes
\langle\tilde {\vec  \beta}|_{G\backslash D}
\rho_G
|\tilde {\vec  \beta}\rangle_{ D}\otimes
|\tilde {\vec \beta}\rangle_{G\backslash D}
\nonumber\\
=&
\sum_{\vec \beta,\vec y\in\mathbb F_q^{N+2n+l}}
q^{|G|-|D'|-N-2n-l}
g(\{\beta_j\}_{j\in D'\cap D})
\langle\tilde {\vec  \beta}|_{ D}\otimes
\langle{\vec  y}|_{G\backslash D}
\rho_G
|\tilde {\vec  \beta}\rangle_{ D}\otimes
|{\vec y}\rangle_{G\backslash D}
\end{align}
holds. First and the last equality just come from the fact that
 both $\{|\tilde \beta\rangle_j \}_\beta$ and $\{\ket{y}_j \}_y$ are bases of the space $\mathcal H_j$.
 The second equality comes from  the property $|\tilde {\vec  \beta}\rangle_{ D}\otimes
|\tilde {\vec \beta}\rangle_{G\backslash D}
=|\tilde {\vec  \beta}\rangle_{ G}$ for any $\vec \beta \in \mathbb F_q^{N+2n+l}$ 
 which derived from  the definition directly. This relation can be used to modify the expression (\ref{eq:g_1})
by substituting $G'$, $G\backslash H$,   
$|M'(\vec a,\vec b,\vec c)^T\rangle_G 
\langle M'(\vec a',\vec b,\vec c')^T|_G$, and 
$\bigotimes_{j\in G'\backslash H\}_{j=1}^h}\ket{\beta_j}_j^{\mathcal V}\bra{\beta_j}_j^{\mathcal V}$
into $D'$, $D$, $\rho_G$, and $g(\{\beta_j\}_{j\in D'\cap D})$ respectively. As a result,  the expression (\ref{eq:g_1}) can be rewrite as
\begin{align}
&q^{-N-3n-n'-2l}\sum_{\vec a,\vec a'\in\mathbb F_q^{n}}
\sum_{\vec b\in\mathbb F_q^{n'}}
\sum_{\vec c,\vec c'\in\mathbb F_q^{h}}
\sum_{\vec \beta,\vec y\in\mathbb F_q^{N+2n+l}}
\ket{\vec a}_I\bra{\vec a'}_I
\nonumber\\
&
\times \langle\tilde {\vec  \beta}|_{G\backslash H}|M'(\vec a,\vec b,\vec c)^T\rangle_{G\backslash H} 
\langle M'(\vec a',\vec b,\vec c')^T|_{G\backslash H}|\tilde {\vec  \beta}\rangle_{G\backslash H}
\nonumber\\
&
\times \langle {\vec  y}|_{H}|M'(\vec a,\vec b,\vec c)^T\rangle_H 
\langle M'(\vec a',\vec b,\vec c')^T|_H| {\vec y}\rangle_{H}
\nonumber\\
&
 \otimes  
(\prod_{j=1}^h\bra{ c_j}_{\varsigma (j)}V_j|M(\vec a,\vec b,\vec c)^T\rangle_{\varsigma (j)})\ket{\phi_{ini}}
\bra{\phi_{ini}}
(\prod_{j=1}^h\bra{ c_j'}_{\varsigma (j)}V_j|M(\vec a',\vec b,\vec c')^T\rangle_{\varsigma (j)})^\dagger
\nonumber\\
&
\otimes \bigotimes_{j\in G'\backslash H}\ket{\beta_j}_j^{\mathcal V}\bra{\beta_j}_j^{\mathcal V}.
\label{eq:G_main_2}
\end{align}
A part of this expression can be evaluated by using the recoverability as follows:
For any $\vec a,\vec a'\in\mathbb F_q^{n}$, $\vec b\in\mathbb F_q^{n'}$, and  $\vec c,\vec c'\in\mathbb F_q^{h}$,
the relation
\begin{align}
&\sum_{\vec y\in\mathbb F_q^{N+2n+l}}
 \langle {\vec  y}|_{H}|M'(\vec a,\vec b,\vec c)^T\rangle_H 
\langle M'(\vec a',\vec b,\vec c')^T|_H| {\vec y}\rangle_{H}
\nonumber\\
=&
q^{N+2n+l-h'}
\langle {M'(\vec a',\vec b,\vec c')^T}|_{H}|M'(\vec a,\vec b,\vec c)^T\rangle_H 
\nonumber\\
=&
q^{N+2n+l-h'}
\langle {\vec  0_{N+2n+l}}|_{H}|M'(\vec a-\vec a',\vec 0_{n'},\vec c-\vec c')^T\rangle_H 
\nonumber\\
=&
q^{N+2n+l-h'}
\delta(\vec 0_{h'},M'_\iota(\vec a-\vec a',\vec 0_{n'},\vec c-\vec c')^T) 
\nonumber\\
=&
q^{N+2n+l-h'}
\delta(\vec a,\vec a')
\delta(\vec 0_{h'},M'_\iota(\vec 0_{n+n'},\vec c-\vec c')^T) 
\label{eq:G_sub_3}
\end{align}
holds where  $h'$ is defined as $|E_P|$ as is done in the case of classical network coding. The first relation justified from the fact that $q^{N+2n+l-|D|}$ is the number of vectors $\vec y\in \mathbb F_q^{N+2n+l}$ which gives an identical state 
by $| {\vec y}\rangle_{D}$ 
 for any $D\subset G$. The second relation comes from the fact that 
$\langle {\vec  y'}|_{D}|\vec y\rangle_D=\langle {\vec 0_{N+2n+l}}|_{D}|\vec  y-\vec y'\rangle_D$ holds for any  $\vec y,\vec y'\in \mathbb F_q^{N+2n+l}$ and $D\subset G$.
The third relation comes from the definition of $H$ and the abbreviation of the computational basis, where 
$M'_\iota$ is made from the $M_0$, $E_A$, and $E_P$ as is done in the case of the classical network coding in the previous section.
The last relation comes form   the recoverability.
That is, if  $M'_\iota(\vec a-\vec a',\vec 0_{n'},\vec c-\vec c')^T=\vec 0_{h'}=M'_\iota\vec 0_{n+n'+h}$, the relation  $\vec a-\vec a'=f_{\vec 0_{n'}}(M'_\iota(\vec a-\vec a',\vec 0_{n'},\vec c-\vec c')^T)=f_{\vec 0_{n'}}(M'_\iota\vec 0_{n+n'+h})=\vec 0_n$ must be hold where $f_{\vec 0_{n'}}$ is the function  defined in  Definition 2. Therefore, the expression (\ref{eq:G_main_2}) becomes
\begin{align}
&
q^{-N-3n-n'-l}
\sum_{\vec a\in\mathbb F_q^{n}}
\ket{\vec a}_I\bra{\vec a}_I
\sum_{\vec b\in\mathbb F_q^{n'}}
\sum_{\vec c,\vec c'\in\mathbb F_q^{h}}
\sum_{\vec \beta\in\mathbb F_q^{N+2n+l}}
\nonumber\\
&
\times 
\omega^{{\rm Tr} \vec \beta_{G\backslash H}^T M'(\vec 0_{n+n'},\vec c'-\vec c)}
\delta(\vec 0_{h'},M'_\iota(\vec 0_{n+n'},\vec c-\vec c')^T) 
\nonumber\\
&
\otimes (\prod_{j=1}^h\bra{ c_j}_{\varsigma (j)}V_j|M(\vec a,\vec b,\vec c)^T\rangle_{\varsigma (j)})\ket{\phi_{ini}}
\bra{\phi_{ini}}
(\prod_{j=1}^h\bra{ c_j'}_{\varsigma (j)}V_j|M(\vec a,\vec b,\vec c')^T\rangle_{\varsigma (j)})^\dagger
\nonumber\\
&
\otimes \bigotimes_{j\in G'\backslash H}\ket{\beta_j}_j^{\mathcal V}\bra{\beta_j}_j^{\mathcal V},
\label{eq:G_main_3}
\end{align}
where $\vec \beta_{G\backslash H}:=(\beta_1',\cdots \beta_{N+2n+l}')$ for  $\beta_j':=\beta_j$ if $j\in G\backslash H$, and  $\beta_j'=0$ if $j\notin G\backslash H$.
Here, in addition to the application of  the relation (\ref{eq:G_sub_3}), we have summed up with respect to $\vec a'$, and we have evaluated the inner product between the computational basis vectors and Fourier basis vectors.
In the next modification, the secrecy for the classical network coding is also used as follows:
Since the corresponding classical network coding is secure, we can define a function $\mathfrak b$ 
which satisfies  the relation (\ref{eq:g_secure_1}). Note that $M_\varsigma$ is uniquely defined from $M_0$ and $M_A$
as is defined in the case of the classical network coding.
Using this function, we can find the relation
\begin{align}
&\sum_{\vec b\in \mathbb F_q^{n'}}g(M_\varsigma(\vec a,\vec b,\vec c),M_\varsigma(\vec a,\vec b,\vec c'))
\nonumber\\
=&\sum_{\vec b\in \mathbb F_q^{n'}}g(M_\varsigma(\vec a,\vec{\mathfrak b}(-\vec a)+\vec b,\vec c),M_\varsigma(\vec a,\vec{\mathfrak b}(-\vec a)+\vec b,\vec c'))
\nonumber\\
=&\sum_{\vec b\in \mathbb F_q^{n'}}g(M_\varsigma(0,\vec b,\vec c),M_\varsigma(0,\vec b,\vec c'))
\end{align}
 for any function  $g$.
The first equality follows from the fact that the set $\mathbb F_q$ is a field, i.e. 
the set $\{x+b|b\in \mathbb F_q\}$ is equal to $\mathbb F_q$ for any $x\in \mathbb F_q$.
In the second equality, we just use the relation (\ref{eq:g_secure_1}).
This relation can be directly applied for the expression (\ref{eq:G_main_3}), i.e. 
$(\prod_{j=1}^h\bra{ c_j}_{\varsigma (j)}V_j|M(\vec a,\vec b,\vec c)^T\rangle_{\varsigma (j)})\ket{\phi_{ini}}
\bra{\phi_{ini}}
(\prod_{j=1}^h\left<\right. c_j'{\left.\right|}_{\varsigma (j)}V_j|M(\vec a,\vec b,\vec c')^T\rangle_{\varsigma (j)})^\dagger
$ is substituted into $g(M_\varsigma(\vec a,\vec b,\vec c),M_\varsigma(\vec a,\vec b,\vec c'))$.
As a result, the expression (\ref{eq:G_main_3}), i.e. the expression (\ref{eq:g_1}), can be evaluated as
\begin{align}
&(q^{-n}
\sum_{\vec a\in\mathbb F_q^{n}}
\ket{\vec a}_I\bra{\vec a}_I
)
\nonumber\\&
\otimes 
(q^{-N-2n-n'-l}
\sum_{\vec b\in\mathbb F_q^{n'}}
\sum_{\vec c,\vec c'\in\mathbb F_q^{h}}
\sum_{\vec \beta\in\mathbb F_q^{N+2n+l}}
\omega^{{\rm Tr} \vec \beta_{G\backslash H}^T M'(\vec 0_{n+n'},\vec c'-\vec c)}
\delta(\vec 0_{h'},M'_\iota(\vec 0_{n+n'},\vec c-\vec c')^T) 
\nonumber\\
&
 \times  
(\prod_{j=1}^h\bra{ c_j}_{\varsigma (j)}V_j|M(\vec 0,\vec b,\vec c)^T\rangle_{\varsigma (j)})\ket{\phi_{ini}}
\bra{\phi_{ini}}
(\prod_{j=1}^h\bra{ c_j'}_{\varsigma (j)}V_j|M(\vec 0,\vec b,\vec c')^T\rangle_{\varsigma (j)})^\dagger
\nonumber\\
&
\otimes \bigotimes_{j\in G'\backslash H}\ket{\beta_j}_j^{\mathcal V}\bra{\beta_j}_j^{\mathcal V}
).
\label{eq:G_main_4}
\end{align}
Here, we have use the fact that $|M(\vec a,\vec b,\vec c)^T\rangle_{\varsigma (j)}$ for any $1\leq j\leq h'$
can be thought as  a function of $M_\varsigma (\vec a,\vec b,\vec c)^T$.
This final expression of the density matrix on $\mathcal H_I\otimes \mathcal W\otimes\mathcal V $ trivially shows that 
the  density matrix is a product state with respect to the partition $\mathcal H_I$ and $\mathcal W\otimes\mathcal V $.

\hfill $\blacksquare$

\section{Examples}\Label{S7}
In this section, we present several examples of secure quantum network codes,
and show their security. 

\subsection{Butterfly network}
We apply Theorem \ref{thm:The-quantum-network main} to the secure network coding of the butterfly network given in our previous paper \cite{OKH}.
The numbers of the edges are assigned as in Fig. \ref{butterfly}.
Almost all the parameters in this case are written down in the following:  
$q$ is a prime power and at the same time it is relatively prime to $2$,
$N=7$, $n=2$, $n'=1$, $l=2$, $l_1=2$,
\begin{align}
\tilde V=&\{v_1,\cdots, v_6\},
\nonumber\\
V_S=&\{v_1,v_2\},
\nonumber\\
V_T=&\{v_5,v_6\},
\nonumber\\
E=&\{\mathbf e(1)=(i_1,v_1),\cdots, \mathbf e(13)=(v_5,o_2)\},
\nonumber\\
 \tilde E=&\{\mathbf e(5)=(v_1,v_3),\cdots, \mathbf e(11)=(v_4,v_6)\},
\nonumber\\
E_I=&\{\mathbf e(1)=(i_1,v_1),\mathbf e(2)=(i_2,v_2)\},
\nonumber\\
E_O=&\{\mathbf e(12)=(v_6,o_1),\mathbf e(13)=(v_5,o_2)\},
\nonumber\\
E_P=&\{\mathbf e(11),\mathbf e(12),\mathbf e(13)\},
\nonumber
\end{align}
\begin{align}
\begin{array}{rclrclrclrcl}
\iota(1)&=&11,
&\iota(2)&=&12,
&\iota(3)&=&13,
\end{array}
\nonumber
\end{align}
\begin{align}
\begin{array}{rclrclrcl}
\inc(5)&=&\{1,3\},
&
\inc(6)&=&\{2,4\},
&
\inc(7)&=&\{1,3\},
\\
\inc(8)&=&\{2,4\},
&
\inc(9)&=&\{5,6\},
&
\inc(10)&=&\{9\},
\\
\inc(11)&=&\{9\},
&
\inc(12)&=&\{8,11\},
&
\inc(13)&=&\{7,10\},
\end{array}
\nonumber
\end{align}
\begin{align}
\begin{array}{rclrclrclrcl}
\theta_{5,1}&=&2,
&
\theta_{5,3}&=&2,
&
\theta_{6,2}&=&2,
&
\theta_{6,4}&=&1,
\\
\theta_{7,1}&=&1,
&
\theta_{7,3}&=&1,
&
\theta_{8,2}&=&1,
&
\theta_{8,4}&=&1,
\\
\theta_{9,5}&=&1,
&
\theta_{9,6}&=&1,
&
\theta_{10,9}&=&1,
&
\theta_{11,9}&=&1,
\\
\theta_{12,8}&=&-1,
&
\theta_{12,11}&=&2^{-1},
&
\theta_{13,7}&=&-1,
&
\theta_{13,10}&=&2^{-1},
\end{array}
\nonumber
\end{align}
\begin{align}
M_0=&\left(\begin{array}{ccccccccccccc}
1&0&0&0&2&0&1&0&2&2&2&1&0\\
0&1&0&0&0&2&0&1&2&2&2&0&1\\
0&0&1&1&1&1&1&1&2&2&2&0&0
\end{array}\right)^T,
\nonumber
\end{align}
\begin{align}
\{v_5,v_6\}=&
\mathfrak E(\mathbf{e}(9))=
\mathfrak E(\mathbf{e}(10))=
\mathfrak E(\mathbf{e}(11)),
\nonumber
\\
\{v_5\}=&
\mathfrak E(\mathbf{e}(2))=
\mathfrak E(\mathbf{e}(6))=
\mathfrak E(\mathbf{e}(8))=
\mathfrak E(\mathbf{e}(13)),
\nonumber
\\
\{v_6\}=&
\mathfrak E(\mathbf{e}(1))=
\mathfrak E(\mathbf{e}(5))=
\mathfrak E(\mathbf{e}(7))=
\mathfrak E(\mathbf{e}(12)),
\nonumber
\\
\phi=&
\mathfrak E(\mathbf{e}(3))=
\mathfrak E(\mathbf{e}(4)).
\end{align}
The additional shared randomness expressed in  Fig. \ref{butterfly} is just used for holding back the measurement outcome $\beta_{11}$.

We assume that Eve attacks only one of edges $\{ \mathbf{e}(5),\cdots, \mathbf{e}(11)\}$. 
As an example, we suppose that $ \mathbf{e}(6)$ is the attacked edge, i.e. $E_A\{\mathbf{e}(6)\}$ and $\varsigma(1)=6$.
In this case, $M$, $M'$, $M_\varsigma$ and $M'_\iota$ can be evaluated as 
\begin{align}
M=&\left(\begin{array}{ccccccccccccc}
1&0&0&0&2&0&1&0&2&2&2&1          &0\\
0&1&0&0&0&2&0&1&0&0&0&        -1&0\\
0&0&1&1&1&1&1&1&1&1&1&2^{-1}-1&2^{-1}-1\\
0&0&0&0&0&0&0&0&1&1&1&2^{-1}   &2^{-1}   
\end{array}\right)^T,
\nonumber\\
M'=&\left(
\begin{array}{ccccccccccccc}
1&0&0&0&2&0&1&0&2&2&2&1          &0\\
0&1&0&0&0&0&0&1&0&0&0&        -1&0\\
0&0&1&1&1&0&1&1&1&1&1&2^{-1}-1&2^{-1}-1\\
0&0&0&0&0&1&0&0&1&1&1&2^{-1}   &2^{-1}   
\end{array}\right)^T,
\nonumber\\
M_\varsigma=&\left(\begin{array}{cccc}
0&2&1&0
\end{array}\right),
\nonumber\\
M'_\iota=&\left(\begin{array}{cccc}
2& 0&1&1\\
1&-1&2^{-1}-1&2^{-1}\\
0&0&2^{-1}-1&2^{-1}
\end{array}\right).
\nonumber\\
\end{align}
By choosing the function $\vec {\mathfrak b}$ as
\begin{align}
\vec {\mathfrak b}(\vec a):=& (2a_2),
\end{align}
 we can check that the condition (\ref{eq:g_secure_1}) holds, i.e. the corresponding classical network coding is secure against Eve's attack on the edge $\{\mathbf{e}(6)\}$. And,
by selecting the function $f_{\vec b}$ as
\begin{align}
f_{\vec b}(\vec y):=& (2^{-1}y_1-y_3-b_1,2^{-1}y_1-y_2-b_1),
\end{align}
we can also check the recoverability for Eve's attack on  $\{\mathbf{e}(6)\}$,
i.e. from the fact,
\begin{align*}
Y_{\iota(1)}=Y_{11}=&2A_1+B_1+C_1\\
Y_{\iota(2)}=Y_{12}=&A_1-A_2+(2^{-1}-1)B_1+2^{-1}C_1\\
Y_{\iota(3)}=Y_{13}=&(2^{-1}-1)B_1+2^{-1}C_1,
\end{align*}
we can check that
\begin{align*}
(A_1,A_2)=&f_{(B_1)}(Y_{11},Y_{12},Y_{13})
\\
=&(2^{-1}Y_{11}-Y_{13}-B_1,2^{-1}Y_{11}-Y_{12}-B_1).
\end{align*}
Therefore, Theorem \ref{thm:The-quantum-network main} guarantees the security of the quantum state
of Protocol \ref{protocol2} against the  attack by Eve on  $\{\mathbf{e}(6)\}$.
 
In fact, even in the case of Eve's attacks on any other edge,
 we can easily show the secrecy   in the classical setting as discussed in \cite{OKH}, and easily check  the recoverability.
Hence, Theorem \ref{thm:The-quantum-network main} guarantees the security of the quantum state
of Protocol \ref{protocol2} against Eve's attack  on any single edge in $\tilde E$.
Indeed, in this case, Protocol \ref{protocol2} is equal to the protocol given in \cite{OKH}.
Therefore, the application of Theorem \ref{thm:The-quantum-network main} 
can be regarded as another proof of the security analysis for the butterfly network given in our previous paper \cite{OKH}. 

\begin{figure}[tbh]
\begin{center}
\includegraphics[scale=0.6]{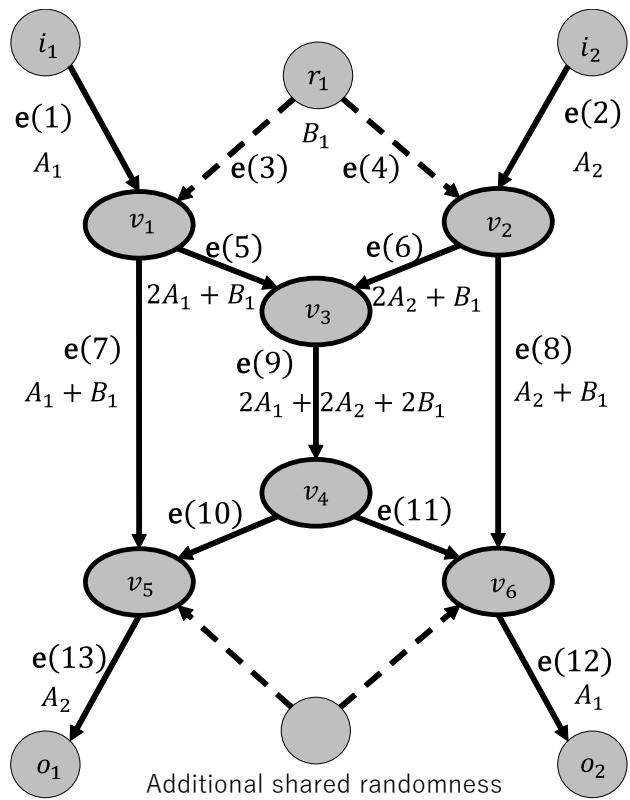}
\caption{Butterfly network with our code
}
\label{butterfly}
\end{center}
\end{figure}

\subsection{Example of networks with $n$-source nodes}\label{subsec one edges}

The  next example is depicted  in Fig. \ref{network1}. 
The  
graph $(\tilde{V},\tilde{E})$ is given as follows.
The set of nodes $\tilde{V}$ is composed of $v_{1}, \ldots, v_{n+2}$, and
the set of quantum channels $\tilde{E}$ is composed of 
$\mathbf{e}(2n+1), \cdots, \mathbf{e}(4n+1)$.   
The vertex $v_j$ is connected to the vertices  $v_{n+1}$ and  $v_{n+2}$
via the edges $\mathbf{e}(2n+j)$ and
 $\mathbf{e}(3n+j)$ respectively where $1\leq j\leq n$.
And, the vertex $v_{n+1}$ is connected to the vertex $v_{n+2}$ via
the edge $\mathbf{e}(4n+1)$.
The source nodes are given as $v_{1}, \ldots, v_{n}$,
and there is single  terminal node  $v_{n+2}$.
Each source node $v_j$ ($1\le j \le n$) 
intends to transmit a $q$-dimensional quantum message to the terminal node $v_{n+2}$,
where  $q$ is a prime power and at the same time it is relatively prime to $n$ and $n-1$.  
And, all the source nodes $v_1, \ldots, v_n$ share one random number $b_1$ of the field $\FF_q$.
Therefore,
 the $n$ input vertices $i_1, \cdots, i_n$ 
are connected to source nodes $v_1, \ldots, v_n$ via 
input edges $\mathbf{e}(1), \cdots, \mathbf{e}(n)$, respectively.
One shared-randomness vertex $r_1$ 
is connected to source nodes $v_1, \ldots, v_n$ via 
shared-randomness edges $\mathbf{e}(n+1), \cdots, \mathbf{e}(2n)$, respectively.
The terminal node $v_{n+2}$
is connected to 
$n$ output vertices $o_1,\cdots, o_n$ 
  via 
the output edges $\mathbf{e}(4n+2), \cdots, \mathbf{e}(5n+1)$, respectively.

The network code $\{ \theta _{j,k} \}_{j\in \{2n+1,\cdots, 5n+1 \}, k\in \inc(j)}$ is defined as follows:
\begin{align}
\begin{array}{rclrclrcl}
\theta_{2n+k,k}&=&n,
&
 \theta_{2n+k,n+k}&=&1,
 \\
   \theta_{3n+k,k}&=&1,
   &
   \theta_{3n+k,n+k}&=&1
 \\
\theta_{4n+1,2n+k}&=&n^{-1},
\\ 
\theta_{4n+k+1,3n+k}&=&1-(n-1)^{-1},
&
 \theta_{4n+k+1,3n+l}&=&-(n-1)^{-1},
 &
  \theta_{4n+k+1,4n+1}&=&(n-1)^{-1},
\end{array}
\Label{eq def network codes 11}
\end{align}
where $1 \le k \le n$, $1 \le l \le n$ and $k \neq l$.
\begin{figure}[tbh]
\begin{center}
\includegraphics[scale=0.38]{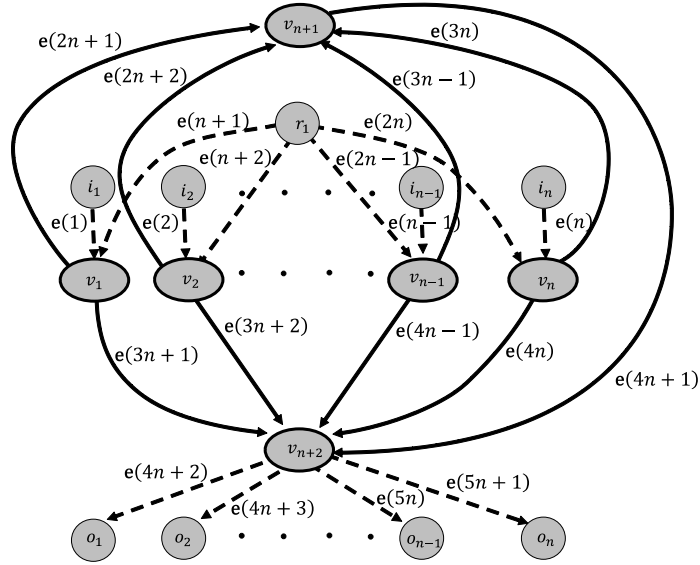}
\caption{The network of the first example, which consists of $n$ input vertices, one shared-randomness vertex,
$n$ output vertices, and $n+2$ nodes
}\Label{network1}
\end{center}
\end{figure}
The set of the protected edges $E_P$ consists of the $n+1$ edges  $\mathbf{e}(3n+1), \cdots, \mathbf{e}(4n+1)$
connecting to the terminal node $v_{n+2}$.
Since all the protected edges connected to the unique terminal node $v_{n+2}$, it is not necessary to send the
 measurement outcomes of the states received from  the channel   $\mathbf{e}(3n+1), \cdots, \mathbf{e}(4n+1)$.
Therefore, we need  not    consume any additional secret randomness in order to hold back the measurements outcomes.

We can easily construct  the $(5n+1)\times (n+1)$ matrix $M_0$ made from $\{\theta_{j,k}\}$, i.e.
\begin{align}
m_0(j,k)=&
\left\{
\begin{array}{cl}
\delta(j,k) & \makebox{if $1\leq j\leq n$ and $1\leq k\leq n+1$} 
\\
\delta(n,k) & \makebox{if $n< j\leq 2n$ and $1\leq k\leq n+1$} 
\\
\delta(j-2n,k)n+\delta(n+1,k) & \makebox{if $2n< j\leq 3n$ and $1\leq k\leq n+1$} 
\\
\delta(j-3n,k)+\delta(n+1,k) & \makebox{if $3n< j\leq 4n$ and $1\leq k\leq n+1$} 
\\
1 & \makebox{if $j=4n+1$ and $1\leq k\leq n+1$} 
\\
\delta(j-4n-1,k) & \makebox{if $4n+1< j\leq 5n+1$ and $1\leq k\leq n+1$},
\end{array}
\right.
\end{align}
and we can check that the condition (\ref{eq:M(n+l+N+c k) for 1<c<n}) satisfies;
that is, we can successfully send $n$ messages  parallelly by the corresponding classical  network code. 
From Theorem \ref{T2-5-1}, that fact  guarantees that the corresponding quantum network code given in Protocol \ref{protocol2} transmits the desired quantum states correctly if there is no attack.

Now, we assume that Eve attacks only one of the edges $\mathbf{e}(2n+1), \cdots, \mathbf{e}(4n+1)$, i.e.  $E_A=\{{\mathbf{e}(j_0)}\}$ for a certain $j_0$ which satisfies $2n+1 \le j_0 \le 4n+1$.
From Theorem \ref{thm:The-quantum-network main}, we know that it is enough to 
 check the secrecy and recoverability of the corresponding classical network codes in order to 
guarantee
the security of the transmitted quantum states, 

From the definition, the $1\times (n+2)$ matrix $M_{\varsigma}$ 
is equal to $(m_0(j_0,1),\cdots,m_0(j_0,n+1),0)$.
Since the matrix $M_{\varsigma}$ have a single raw and the $n+1$-th column of the matrix is non-zero,
 we can construct the function $\vec {\mathfrak b}$ which satisfies the relation (\ref{eq:M(n+l+N+c k) for 1<c<n}), i.e. 
the corresponding  classical network code is 
  secure against Eve's attack on the edge $\{\mathbf{e}(j_0)\}$.

The recoverability of the corresponding classical network code is shown as follows.
When $E_A=\{{\mathbf{e}(3n+k)}\}$ with $1\le k \le n$, 
\begin{equation}
M_{\iota}\left( 
\vec{a},
0,
c_1
\right)^T
=( a_1,\cdots,a_{k-1}, c_1, a_{k+1}\cdots,a_n, \sum_{i=1}^n a_i ).
\end{equation}
When $E_A=\{{\mathbf{e}(2n+k)}\}$ with $1\le k \le n$, 
\begin{equation}
M_{\iota}\left( 
\vec{a},
0 ,
c_1
\right)^T
=( a_1,\cdots,a_n, \sum_{i=1}^n a_i -a_k +c_1).
\end{equation}
When $E_A=\{{\mathbf{e}(4n+1)}\}$, 
\begin{equation}
M_{\iota}\left( 
\vec{a},
0 ,
c_1
\right)^T
=( a_1,\cdots,a_n, c_1).
\end{equation}
In every case, it is easy to check that there exists a function $f_0$ 
satisfying (\ref{eq:sec classical subsec robustness f a M i a e =00003D i}). The existence is  equivalent to 
meeting the second condition in  Lemma \ref{thm:classical robust main} which is  given in Appendix \ref{AL} holds. Therefore,
 the classical network code is recoverable for any  Eve's attacks on any single communication channel in $\tilde{E}$. 

Therefore,  Theorem \ref{thm:The-quantum-network main} indicates that  the quantum network coding given by Fig. \ref{network1} with  (\ref{eq def network codes 11}) is secure for any Eve's attacks on any single 
quantum  channel in $\tilde{E}$.

\subsection{Network that is secure against all attacks on any two edges}\label{subsec two edges}

The network of the next example is shown in Fig. \ref{network2}. 
The corresponding  
graph $(\tilde{V},\tilde{E})$ is formally given as follows.
The set of nodes $\tilde{V}$ is composed of $v_{1}, \ldots, v_{5}$, and
the set of quantum channels $\tilde{E}$ is composed of 
$\mathbf{e}(7), \cdots, \mathbf{e}(14)$.   
$v_1$ is connected to $v_3$, $v_4$, and $v_5$ via 
$\mathbf{e}(7)$,
$\mathbf{e}(9)$, and
$\mathbf{e}(11)$ respectively.
$v_2$ is also connected to $v_3$, $v_4$, and $v_5$ via 
$\mathbf{e}(8)$,
$\mathbf{e}(10)$, and
$\mathbf{e}(12)$.
And,
$v_5$ is additionally connected from $v_3$ and $v_4$ via  
$\mathbf{e}(13)$ and
$\mathbf{e}(14)$.

The source nodes are given as $v_{1}, v_{2}$
and the terminal node is given as $v_{5}$.
Source nodes $v_1$ and $v_2$ 
intend to transmit a $q$-dimensional quantum message to terminal node $v_{5}$,
where we assume that $q$ is relatively prime to $2$, $3$, and $5$.
In this network, all source nodes $v_1, v_2$ share two random numbers $b_1$ $b_2$ of the finite field $\FF_q$.
As a result, the two input vertices $i_1, i_2$ 
are connected to source nodes $v_1, v_2$ via 
input edges $\mathbf{e}(1),  \mathbf{e}(2)$, respectively.
A shared-randomness vertex $r_1$ ($r_2$) 
is connected to source nodes $v_1$, $v_2$ via 
the shared-randomness edges $\mathbf{e}(3),\mathbf{e}(4)$ ($\mathbf{e}(5),\mathbf{e}(6)$), respectively.
The two output vertices $o_1, o_2$ 
are connected from terminal node $v_{5}$ via 
output edges $\mathbf{e}(15),  \mathbf{e}(16)$, respectively.

Then, the network code is defined by the following parameters:
\begin{align}
\begin{array}{rclrclrclrcl}
\theta_{7,1}&=&1,&
\theta_{7,3}&=&1,&
\theta_{7,5}&=&0,\\
\theta_{9,1}&=&1,&
\theta_{9,3}&=&1,&
\theta_{9,5}&=&1,\\
\theta_{11,1}&=&1,&
\theta_{11,3}&=&0,&
\theta_{11,5}&=&1,\\
\theta_{8,2}&=&1,&
\theta_{8,4}&=&2,&
\theta_{8,6}&=&1,\\
\theta_{10,2}&=&2,&
\theta_{10,4}&=&1,&
\theta_{10,6}&=&2,\\
\theta_{12,2}&=&1,&
\theta_{12,4}&=&1,&
\theta_{12,6}&=&3,\\
\theta_{13,7}&=&1,&
\theta_{13,8}&=&1,\\
\theta_{14,9}&=&1,&
\theta_{14,10}&=&1,\\
\theta_{15,11}&=&3\times 4^{-1},&
\theta_{15,12}&=&- 2^{-1},&
\theta_{15,13}&=&0,&
\theta_{15,14}&=&4^{-1},\\
\theta_{16,11}&=&-5\times 8^{-1},&
\theta_{16,12}&=&-3\times 4^{-1},&
\theta_{16,13}&=&-2^{-1},&
\theta_{16,14}&=&9\times 8^{-1},
\label{eq def network codes 111}
\end{array}
\end{align}
\begin{figure}[tbh]
\begin{center}
\includegraphics[scale=0.40]{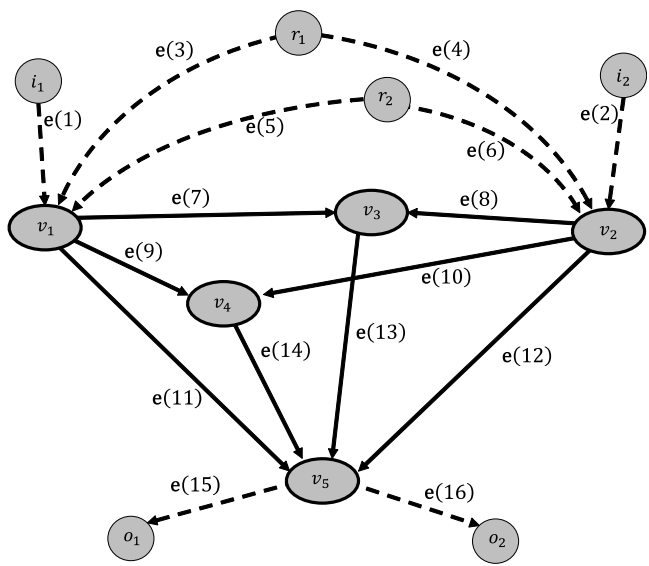}
\caption{The network of the first example, which consists of $n$ input vertices, one shared-randomness vertices,
$n$ output vertices, and $n+2$ nodes
}\label{network2}
\end{center}
\end{figure}
The set of the protected edges $E_P$ consists of  the four edges $\mathbf{e}(11), \mathbf{e}(12),\mathbf{e}(13), \mathbf{e}(14)$ connecting to terminal node $v_{5}$.
Since this network has the single terminal node $v_{5}$, it is not necessary to send 
all the measurement outcomes from edges $\mathbf{e}(11), \mathbf{e}(12),\mathbf{e}(13), \mathbf{e}(14)$.
Thus, we need not consume any additional secret randomness to hold back the measurement outcomes.

By straightforward calculations, we can check that the network code satisfies condition (\ref{eq:M(n+l+N+c k) for 1<c<n}); therefore, we can successfully send $2$ characters parallelly with the corresponding classical network code. 
That is, Theorem \ref{T2-5-1} guarantees that the corresponding quantum network code given in Protocol \ref{protocol2}
transmits the desired quantum messages correctly if there is no attack on all the edges.

Now, we assume that Eve attacks any two of edges in the set $\tilde E$; $E_A=\{ \mathbf{e}(j_0),\mathbf{e}(k_0) \}$ for $7 \le j_0< k_0 \le 14$.
From Theorem \ref{thm:The-quantum-network main},
we can guarantee
the security of the transmitted quantum message by 
 checking the secrecy and recoverability of the corresponding classical network codes. 

This network coding satisfies $n'=h=2$.
We can directly calculate $M_{\varsigma,2}$ 
 and verify that $M_{\varsigma,2}$ is an invertible matrix for any choice of $E_A=\{ \mathbf{e}(j_0),\mathbf{e}(k_0) \}$ with $7 \le j_0< k_0 \le 14$. For example, in the case of $j_0=8$, and $k_0=13$,
we can evaluate  $M_{\varsigma,2}$ as $\left(
 \begin{array}{cc}
2&1\\
1&0
\end{array}
 \right)$. 
Thus, Corollary \ref{cor security classical code 1} guarantees the secrecy
of this classical network code against Eve's attack. 

We next focus on the recoverability of the corresponding classical network code. 
From the second condition in  Lemma \ref{thm:classical robust main}
proved in the Appendix \ref{AL}, we only need to consider the case where all random variables are fixed to $0$. In this case the information on the edges on $\mathbf{e}(11), \mathbf{e}(12), \mathbf{e}(13)$, and $\mathbf{e}(14)$ can be written as $A_1$, $A_2$, $A_1+A_2$, and $A_1+2A_2$, respectively, where $A_1$ and $A_2$ are the information sent from $I_1$ and $I_2$, respectively if there are no disturbances. 
Hence, we can recover $A_1$ and $A_2$ from any two of the edges. 
Now, from the topology of the graph, Eve's attack on $E_A=\{\mathbf{e}(j_0),\mathbf{e}(k_0)\}$ affects at most two of these edges. 
Therefore, the protected edges $E_P$ including the above edges are recoverable.

Finally, from Theorem \ref{thm:The-quantum-network main}, the quantum network coding given by Fig. \ref{network2} with  Eq.(\ref{eq def network codes 111}) is secure for all Eve's attack on the any two of
quantum  channels in $\tilde{E}$.

\subsection{Quantum threshold ramp secret sharing}
Quantum secret sharing (QSS) \cite{CGL} is a protocol to encrypt a quantum state into a multipartite state
so that each system (share) has no information and 
an original state can be reproduced from a collection of the systems. 
Various different QSS schemes have been developed \cite{CGL,G00,OSIY2005,DLLZZ,MS08,YTCW}.
Among them, a $(k,L,n)$-threshold ramp QSS scheme is defined as
a QSS scheme with $n$ shares having the following property \cite{OSIY2005}: The original state can be reconstructed from 
any $k$ shares, and  any $k-L$ shares has no information. 
Hence, partial information of the original state can be drived from $t$ shares with $k > t >k-L$.  
The network codes given in the above subsections $B$ and $C$ are strongly related to
$(k,L,n)$-threshold ramp QSS scheme with $k=n$. 
Here, the condition $k=n$ means that all the $n$ shares are required to reconstruct the original state. 

\begin{figure}[tbh]
\begin{center}
\includegraphics[scale=0.38]{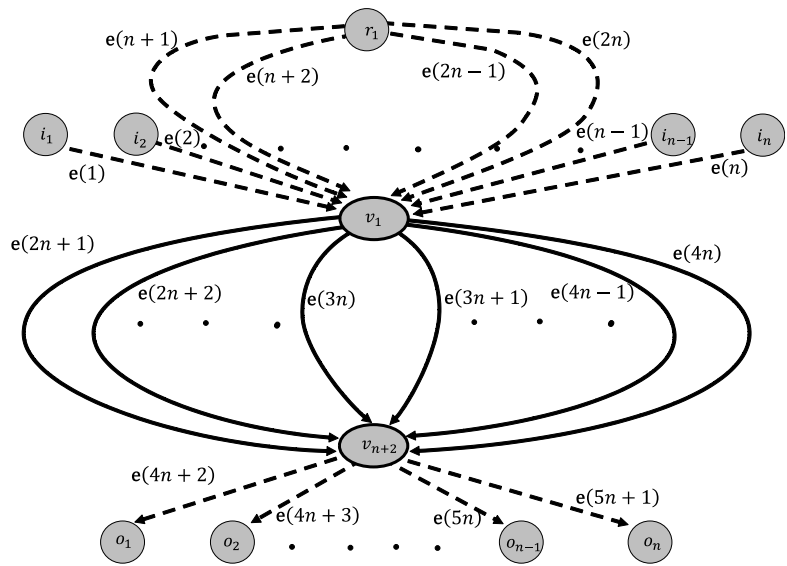}
\caption{The network that can be derived from the second example by contracting the edge $\mathbf{e} (4n+1)$
and merging vertices $\mathbf{v}_i$ with $1\le i \le n$ in graph theoretical sense \label{network secret sharing 1}}
\end{center}
\end{figure}

The network code given in the subsection $B$ is related to
a $(2n,2n-1,2n)$-threshold ramp QSS scheme. Let us consider a new network in Fig. \ref{network secret sharing 1} which can be derived from the network in Fig. \ref{network1} by the following modification of the graph. 
$n-1$ vertices from $v_{2}$ to $v_{n}$ are also merged into the vertex $v_1$, i,e, a set of vertices $\{v_i \}_{1\le i \le n}$ are replaced by a single vertex  $v_{1}$. 
The vertex 
$v_{n+1}$ is also merged into the vertex $v_{n+2}$.  As a result,
the edge $\mathbf{e}(4n+1)$  disappears. 
All the edges connected to an old replaced vertex are connected to the corresponding new 
vertex, and all the edges connected from an old replaced vertex are 
connected from the corresponding new vertex.  
Following this modification, the network code is also modified as follows: 
\begin{align}
\begin{array}{rclrclrcl}
\theta_{2n+k,k}&=&n,
&
 \theta_{2n+k,n+k}&=&1,
 \\
   \theta_{3n+k,k}&=&1,
   &
   \theta_{3n+k,n+k}&=&1
\\ 
\theta_{4n+k+1,3n+k}&=&1-(n-1)^{-1},
&
 \theta_{4n+k+1,3n+l}&=&-(n-1)^{-1},
 &
  \theta_{4n+k+1,2n+l}&=&n^{-1}(n-1)^{-1},
\end{array}
\end{align}
where $1 \le k \le n$, $1 \le l \le n$ and $k \neq l$. Note that, the indexing of the vertices $v_j$ and edges $\mathbf e(j)$ breaks  the general description rule defined in the previous section in order to make it easy to compare  this example and that in the subsection B.
From the security analysis of the subsection B, 
this network code, which does not have any intermediate nodes, is apparently secure against Eve's attack on any one of the $2n$ channels.
On the other hand, all the information on $2n$ channels are required to recover the original quantum state.
Further, since the classical randomness is used only in $v_1$, 
the classical randomness can be generated on  the node $v_1$. 
Hence, as a protocol sending $n$-quantum messages from the input node $v_1$ to the output node $v_{n+2}$,
this network coding is nothing but $(2n,2n-1,2n)$ quantum threshold ramp secret sharing scheme \cite{OSIY2005}.   

\begin{figure}[tbh]
\begin{center}
\includegraphics[scale=0.38]{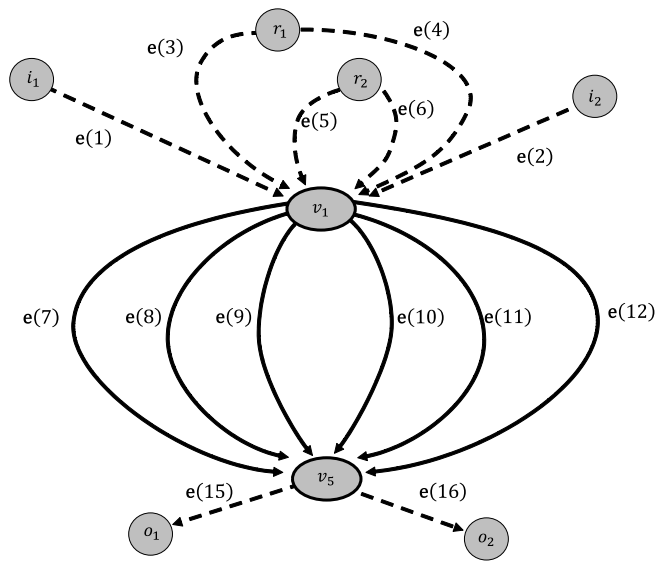}
\caption{The network that can be derived from the third example by contracting the edge $\mathbf{e} (4n+1)$
and merging vertices $\mathbf{v}_i$ with $1\le i \le n$ in graph theoretical sense \label{network secret sharing 2}}
\end{center}
\end{figure}

The network code given in the subsection $C$ is also related to a $(6,4,6)$ quantum ramp secret sharing scheme. 
Let us consider a new network in Fig. \ref{network secret sharing 2} which can be derived from the network in Fig. \ref{network2} 
by the following modification of the graph operations. 
The vertex $v_{2}$ is merged into the vertex $v_{1}$.
The vertices $v_{3}$ and $v_{4}$ are also merged into the vertex $v_5$.
  As a result, the edges $\mathbf{e}(13)$ and  $\mathbf{e}(14)$
  disappears. 
All the edges connected to an old replaced vertex are connected to the corresponding new 
vertex, and all the edges connected from an old replaced vertex are 
connected from the corresponding new vertex.  
Following this modification, the network code is also modified as follows: 
\begin{align}
\begin{array}{rclrclrclrcl}
\theta_{7,1}&=&1,&
\theta_{7,3}&=&1,&
\theta_{7,5}&=&0,\\
\theta_{9,1}&=&1,&
\theta_{9,3}&=&1,&
\theta_{9,5}&=&1,\\
\theta_{11,1}&=&1,&
\theta_{11,3}&=&0,&
\theta_{11,5}&=&1,\\
\theta_{8,2}&=&1,&
\theta_{8,4}&=&2,&
\theta_{8,6}&=&1,\\
\theta_{10,2}&=&2,&
\theta_{10,4}&=&1,&
\theta_{10,6}&=&2,\\
\theta_{12,2}&=&1,&
\theta_{12,4}&=&1,&
\theta_{12,6}&=&3,\\
\theta_{15, 9}&=&4^{-1}&
\theta_{15, 10}&=&4^{-1},&
\theta_{15,11}&=&3\times 4^{-1},&
\theta_{15,12}&=&- 2^{-1},\\
\theta_{16,7}&=&-2^{-1},&
\theta_{16,8}&=&-2^{-1},&
\theta_{16,9}&=&9\times 8^{-1},&
\theta_{16,10}&=&9\times 8^{-1},\\
\theta_{16,11}&=&-5\times 8^{-1},&
\theta_{16,12}&=&-3\times 4^{-1},
\end{array}
\end{align}
From the security analysis of the subsection C, 
 this new network code, which does not have any intermediate nodes, is apparently secure against Eve's attack on any two of the $6$ channels.
On the other hand, all the information on $6$ channels are required to recover the original quantum state. 
Further, since the classical randomness is used only in $v_1$, 
the classical randomness can be generated on  the node $v_1$. 
Hence, as a protocol sending a quantum message from the input node $v_1$ to the output node $v_{5}$,
this network coding is nothing but $(6,4,6)$ quantum threshold ramp secret sharing scheme \cite{OSIY2005}.   

\section{
Advantages of our quantum network code against
 quantum error correcting code on partially corrupted quantum network
}
In this paper, we give a way to make protocols of secure transfer of quantum messages on quantum networks designed originated from classical network coding. 
However, it has been already investigated to construct such a protocol designed originated from quantum error correcting code, i.e.  quantum error correcting code on partially corrupted quantum network~\cite{SH18-1,SH18-2,HS19}.
Therefore, we think that it is fair to compare the secure quantum network coding given in this paper and the quantum error correcting code on partially corrupted quantum network.

As a special property of quantum information, it is well known that, if quantum messages can be transferred with fidelity $1$, it is guaranteed that any other party can't get 
any information about the quantum messages. 
Therefore, it is natural to apply this property to construct protocols of secure transfer of quantum messages on  quantum network which is made from the following three processes. 
1) By using a quantum error correcting code, a quantum message is encoded into several quantum characters at the  source nodes.
2) The quantum characters are sent to terminal nodes via a quantum network.
3) At the terminal nodes, 
the transmitted quantum characters are decoded into the original quantum message.
If the amount of disturbances by Eve is bounded by a threshold given by the error correcting code, the secrecy  and reliability of the transfer of the message are simultaneously guaranteed.
Such an idea has been discussed by several papers~\cite{SH18-1,SH18-2,HS19}. However, 
our construction of the quantum network coding has two advantages against these previous works.

First advantage is a wide applicability. 
Even in the previous papers~\cite{SH18-1,SH18-2,HS19}, operations in the intermediate nodes are designed originated from classical network coding automatically.
However,  all the operations on the intermediate node are restricted to be unitary operations.  
For example, all the node operations are  quantum unitary gates  designed originated from arbitrary bijective
linear maps~\cite{SH18-2}.
As a result, only the bijective functions can be used to design the quantum operators. Strictly speaking, only the invertible  functions can be used.
From this restriction, 
we can't construct a quantum network protocol by  simple application of quantum error correcting code even on the butterfly network for example. 
 Therefore,
very restricted types of quantum network protocols can be constructed from the previous papers especially in the sense of the variety on the intermediate nodes. 
In the case of quantum network coding in this paper,
the operations in the intermediate nodes are CPTP map generally, i.e. unitary operations and measurement operations.
As a result, 
we can design the node operations originated even from irreversible linear maps. Note that such a property 
  is inherited from the previous result regarding  the construction of quantum network coding designed originated from classical network coding without secrecy~\cite{Kobayashi2010} which is a basis of our result.

Second advantage is an improvement of the secrecy.
As we mentioned, in the quantum network protocol made from quantum error correcting code, the secrecy and reliability is indistinguishable.
As a result, the secrecy of the code is deeply connected to theoretical limits of quantum error correction.
However, in the quantum network coding proposed here, 
even if the terminal node can't recover the original quantum message,
it is possible that the two conditions in Theorem 2  hold with
respect to the set $E_P$ of the protected edges.
In this case, the secrecy of the quantum message is guaranteed\footnote{The reference \cite{ICITS} showed that this condition is equivalent to the recoverability of the original quantum message by collecting the
information from all the protected edges.}.
Therefore, the secrecy is not necessarily restricted by theoretical limits of error correcting code.

\section{Conclusion}\Label{S5}
Based on a secure classical network code, we have proposed a canonical way to make a secure quantum network
code in the multiple-unicast setting. This protocol certainly transmits quantum states when there is no attack.
While our protocol needs classical communications, they are limited to one-way communications
i.e.,
 all of the classical information is given by predefined measurements on nodes and only the final operators on the terminal nodes are affected by the information.
Hence, it does not require verification process, which ensures single-shot security. We have also shown the secrecy of the quantum network code under the secrecy and the recoverability of the
corresponding classical network code.
Our security proof focuses on the classical recoevrbility and the classical secrecy \cite{Renes}.

Our protocol offers secrecy different from that of QKD.
While our protocol has the restriction of the number of attacked edges, our protocol does not require repetitive quantum communications
because it does not need a verification process. In contrast, QKD needs repeatative quantum communications,
which enables us to verify the non-existence of the eavesdropper and to ensure the security. Finally, although the previous result \cite{OKH} can be applied only to a special secure code on the butterfly network,
our secure network code can be applied to any secure classical network code.
We have demonstrated several application of our code construction in various network including the butterfly network.
These applications show applicability of our method.

\section*{Acknowledgments}
The authors are very grateful to 
Professor Ning Cai and Professor Vincent Y. F. Tan
for helpful discussions and comments.
The works reported here were supported in part by 
the JSPS Grant-in-Aid for Scientific Research 
(C) No. 16K00014, (B) No. 16KT0017,  
(C) No. 17K05591, (A) No. 23246071, 
the Okawa Research Grant
and Kayamori Foundation of Informational Science Advancement.

\appendices

\section{Lemmas for classical network code}\Label{AL}
We give a corollary and  a lemma for classical network coding
that are used for the analysis on our examples given in the section \ref{S7}.

\subsection{Corollary for secrecy}

We can obtain the following corollary of Lemma \ref{L88}, which is useful for actual analysis.
\begin{corollary}\Label{cor security classical code 1}
When $h=n'$, a (classical) network code is secure for all of Eve's attacks on
$E_{A}$, if $M_{\varsigma,2}$ is invertible. 
In particular, when  $h=n'=1$, the (classical) network code is secure for all Eve's attack on
$E_{A}$, if $M_{\varsigma,2} \neq 0$.
\end{corollary}
\begin{proof}
When $M_{\varsigma,2}$ is invertible,
$M_{\varsigma,2}$ is surjective.
Thus, the image of $M_{\varsigma,1}$ is contained in that of $M_{\varsigma,2}$.

When $h=n'=1$, $M_{\varsigma,2}$ is just an element of a finite field. 
Hence, it is invertible if and only if it is non-zero.
\end{proof}

\subsection{Lemma for recoverability}
We can relax the recoverability condition from  Definition \ref{DD5} as follows:
\begin{lemma}
\Label{thm:classical robust main}
The following three conditions are equivalent: 
\begin{enumerate}
\item 
The messages are recoverable for Eve's attack on $E_A$ by  $E_P$. 

\item 
There exists a function $f_{\vec 0_{n'}}:\mathbb{F}_{q}^{\left|\Ed\right|}\rightarrow\mathbb{F}_{q}^{n}$
satisfying 
\begin{equation}\Label{eq robustness 11}
f_{\vec 0_{n'}}\left(M'_{\iota} \cdot\left(
\vec{a},
\vec{0}_{n'}
\vec{c}
\right)^T\right)=\vec{a}
\end{equation}
for all $\vec{a}\in\mathbb{F}_{q}^{n}$ and $\vec{c}\in\mathbb{F}_{q}^{h}$.

\item 
There exist an $n$-by-$\left|\Ed\right|$ matrix $\mathfrak M_1$ and an $n$-by-$n'$
matrix $\mathfrak M_2$ such that the relation 
\begin{equation}
\mathfrak M_1 \cdot M'_{\iota}\cdot\left(
\vec{a},
\vec{b},
\vec{c}
\right)^T
=\vec{a}+\mathfrak M_2 \cdot\vec{b}
\Label{eq:sec classical subsec robustness Y M iae =00003D i+za}
\end{equation}
holds for any vectors $\vec{a}\in\FF_{q}^{n}$, 
$\vec{b} \in\FF_{q}^{n'}$,
and $\vec{c}\in\FF_{q}^{h}$.

\end{enumerate}
\end{lemma}

The last condition in this lemma means that 
if there exists a decoder, it can be always chosen as a linear decoder.

\begin{IEEEproof}
Since the directions 3)$\Rightarrow$1)$\Rightarrow$2) is trivial, we show only 
2)$\Rightarrow$3).

Assume 2).
We easily find that $f_{\vec 0_{n'}}$ can be restricted to be linear since the condition (\ref{eq robustness 11}) demands the function $f_{\vec 0_{n'}}$ to be linear on the region expressed by the form $M'_{\iota} \cdot(
\vec{a},
\vec{0}_{n'}
\vec{c}
)^T$  for any vectors $\vec{a}\in\FF_{q}^{n}$, 
and $\vec{c}\in\FF_{q}^{h}$.
Hence, $f_{\vec 0_{n'}}$ on the image of $M'_{\iota}$ can be written as an $n$-by-$\left|\Ed\right|$ matrix $\mathfrak M_1$.
Since the map
$ \vec{b}\mapsto
\mathfrak M_1\cdot M'_{\iota}\cdot\left(
\vec 0_n,
\vec{b},\vec 0_h
\right)^T$ 
is linear, 
there exists an $n$-by-$n'$ matrix $\mathfrak M_2$ such that 
$\mathfrak M_2 \vec{b}=\mathfrak M_1\cdot M'_{\iota}\cdot
\left(
\vec 0_n,
\vec{b},\vec 0_h
\right)^T$. 
Thus,
\begin{align*}
\mathfrak M_1\cdot M'_{\iota}\cdot\left(
\vec{a},
\vec{b},
\vec{c}
\right)^T
=&
\mathfrak M_1\cdot M'_{\iota}\cdot\left(
\vec{a},
\vec 0_{n'},
\vec{c}
\right)^T
+
\mathfrak M_1\cdot M'_{\iota}\cdot\left(
\vec 0_n,
\vec{b},\vec 0_h
\right)^T \\
=&\vec{a}+\mathfrak M_2 \vec{b},
\end{align*}
which implies 3).
\end{IEEEproof}

\section{Constructions of matrices describing network}\Label{AConst}
In this appendix, we concretely construct the matrices describing the network structure.

\subsection{Construction of $M_0$}\Label{a1}

The definition of
input edges and shared-randomness edges determine the coefficients
$\left\{ m_0\left(j,k\right)\right\} _{j,k}$ for $1\le j\le n+l$
as follows: For $1\le j\le n$, $\mathbf{e}\left(j\right)$ is an
input edge, that is, $\mathbf{e}\left(j\right)\in E_{I}$.
Thus, the definition of input edges determines $\left\{ m_0\left(j,k\right)\right\} _{k=1}^{n+n'}$ as 
\begin{equation}
\left\{ m_0\left(j,k\right)\right\} _{k=1}^{n+n'}=(\vec 0_{j-1},1,\vec 0_{n+n'-j})\quad\mbox{for }1\le j\le n.\Label{eq:M-ik for 1<i<n}
\end{equation}

For $n+1\le j\le n+l$, $\mathbf{e}\left(j\right)$ is a shared-randomness
edge, that is, $\mathbf{e}\left(j\right)\in E_{R}$. 
Hence, there uniquely exists an integer $j' \in [1,n']$ such that
$n+\sum_{k'=1}^{j'-1}l_{k'}<  j\le n+\sum_{k'=1}^{j'}l_{k'}$.
Thus, the definition of shared-randomness edges determines 
$\left\{ m_0\left(j,k\right)\right\}_{k=1}^{n+n'}$ as
\begin{align}
\left\{ m_0\left(j,k\right)\right\} _{k=1}^{n+n'} 
& =(\vec 0_{n+i'-1},1,\vec 0_{n'-i'})
\quad\mbox{for }n< j\le n+l.
\Label{eq:M-ik for n<i<n+l}
\end{align}
By substituting the expression \eqref{H204} of $Y_j$  into 
the relation 
 \eqref{eq:y(i)=00003Dsum a-ij y(j)},
 we derive the recurrence relation
of $m_0\left(j,k\right)$ as
\begin{equation}
m_0\left(j,k\right)=\sum_{k'<j}\theta_{j,k'}m_0\left(k',k\right)
\Label{eq:M-ik=00003Dsum a-ij M-jk}
\quad\mbox{for }n+l< j\le N+2n+l.
\end{equation}
Note that $M_0$ is a matrix which identifies the relation between the  character transferred on the edges and the combination of  messages and shared-secure-random number  in the case that there is no disturbance for every channel.
Therefore, we can use  Eq.\eqref{eq:y(i)=00003Dsum a-ij y(j)}
and \eqref{H204}.

The Eqs. \eqref{eq:M-ik for 1<i<n},\eqref{eq:M-ik for n<i<n+l},
and \eqref{eq:M-ik=00003Dsum a-ij M-jk} enable us 
to evaluate  all the coefficients of the $(N+2n+l)\times(n+n')$ matrix $M_0$, i.e. $\left\{ m_0\left(j,k\right)\right\} _{j,k}$,
recursively.

\subsection{Construction of $M$}\Label{a2}

In the case of $1\leq  j \leq n+l$, $Y_j$ is not affected by disturbances by definition. 
Therefore, 
\begin{eqnarray}
m\left(j,k\right)&=&
\left\{
\begin{array}{cl}
m_0\left(j,k\right)& \makebox{for $1\leq  j \leq n+l$ and $1\leq k\leq n+n'$} 
\\
0& \makebox{for $1\leq  j \leq n+l$ and $n+n'< k\leq n+n'+h$} 
\end{array}
\right.
\label{eq:G_tmp_1}
\end{eqnarray}

$M$ is a matrix which identifies the relation between the  character transferred on the edges and the combination of  messages, shared-secure-random number and injected character.
That means, we consider the case that there may exist disturbances. Therefore,
we have to use the relation 
\[
Y_j =\sum_{k<j}\theta_{j,k}
Y_k'
\quad\mbox{for }n+l< j\le N+2n+l
\]
instead of the the relation Eq.\eqref{eq:y(i)=00003Dsum a-ij y(j)}.
By substituting the expressions  \eqref{eq:sec classical subsec security y i =00003D M i k x k}
and \eqref{eq:def y'}  of $Y_j$ and $Y_j'$ into the above relation,
we obtain the relation
\begin{equation}
m\left(j,k\right)=\sum_{k'<j}\theta_{j,k'}m'(k',k)
\quad\mbox{for }n+l< j\le N+2n+l.
\Label{eq:M ik=00003D sum theta ij M jk}
\end{equation}
By combining  \eqref{eq:M' ik =00003D delta M ik} for the above relation, 
we derive the following recurrence relations for $m\left(j,k\right)$:
\begin{equation}
m\left(j,k\right)=\sum_{k'<j}
\theta_{j,k'}m\left(k',k\right)+\sum_{k'=1}^{h}\theta_{j,\varsigma\left(k'\right)}(\delta_{k,n+n'+k'}-m\left(\varsigma\left(k'\right),k\right))\Theta(j-\varsigma\left(k'\right)-1),
\Label{eq:M ik=00003Dsum theta ij M jk sum theta delta}
\end{equation}
for $n+l< j \leq N+2n+l$,
where $\Theta(y)$ is a step function such that $\Theta(y)=0$ ($\Theta(y)=1$) if $y<0$ ($y\geq 1$).

The Eqs. \eqref{eq:G_tmp_1} and \eqref{eq:M ik=00003Dsum theta ij M jk sum theta delta} enable us 
to evaluate  all the coefficients of the $(N+2n+l)\times(n+n'+h)$ matrix $M$ recursively.

\end{document}